\documentclass[11pt]{article}

\usepackage{amssymb}
\usepackage{graphicx}
\usepackage{latexsym}
\usepackage{mathrsfs}  
\usepackage{array}
\usepackage{graphicx}
\usepackage{color}
\usepackage{url}
 \usepackage{amsthm}
  \newtheorem{example}{Example}
    \newtheorem{theorem}{Theorem}
    
\begin{document}




\title{Derivation languages, descriptional complexity measures and decision problems  of a class of flat splicing systems}


\author{Prithwineel Paul\\
Electronics and Communication Sciences Unit\\
Indian Statistical Institute, 203 B. T. Road Kolkata,\\
West Bengal 700108, India
\and
Kumar Sankar Ray\\
Electronics and Communication Sciences Unit\\
Indian Statistical Institute, 203 B. T. Road Kolkata,\\
West Bengal 700108, India}

\author{ Prithwineel Paul*, Kumar Sankar Ray}

\maketitle

\begin{abstract}
In this paper, 
we associate the idea of derivation languages with flat splicing systems
and compare the families of derivation languages (Szilard and control languages) of these systems
with the family of languages in Chomsky hierarchy. 
We show that the family of Szilard languages of labeled  flat finite splicing systems of type $( m, n)$ (i.e., $SZLS_{n, FIN}^{m}$ ) and $REG$, $CF$ and $CS$
are incomparable.
Also,  it is decidable whether or not  $SZ_{n, FIN}^{m}(\mathscr{LS}) \subseteq R$ and $R \subseteq SZ_{n, FIN}^{m}(\mathscr{LS})$
for any regular language $R$ and labeled flat  finite splicing system $\mathscr{LS}$.  
 Also, any non-empty regular, non-empty context-free and recursively enumerable language 
can be obtained as homomorphic image of Szilard language of the labeled flat  finite splicing
 systems of type $(1, 2), (2, 2)$ and $(4, 2)$ respectively. 
We also introduce the idea of control languages for
  labeled  flat finite splicing systems and show that
any non-empty regular and context-free language can be obtained 
as a control language of labeled  flat  finite splicing systems of type $(1,2)$ and $(2, 2)$ respectively.
 At
 the end, we show that any recursively enumerable language 
can be obtained as a control language of labeled flat  finite splicing systems of type $(4,2)$ when $\lambda$-labeled rules
 are allowed.  

\end{abstract}





\textbf{keyowrds:}Flat Splicing, Labeled flat splicing systems, Szilard language, Control language, Chomsky hierarchy



\section{Introduction} \label{Introduction}
\noindent Splicing systems mathematically formalize recombination behavior of DNA 
molecules under
 the presence of restriction enzymes
and ligases. The restriction enzymes cut the DNA molecules in specific cites and ligases 
help the molecules to recombine with each other to produce new molecules. After each splicing step new DNA molecules and sometimes the same molecules taking 
part in splicing are produced.
Splicing systems are well-investigated topic in formal language theory \cite{dnabook}. Different variants of splicing systems and their computational capabilities already have been 
investigated in \cite{fkp,dsp,tp,dnabook}. 
Splicing systems containing finite set of axioms and rules can generate only regular languages \cite{culik}. But with different restrictions
on the set of axioms and the rules, finite splicing systems can even characterize recursively enumerable languages \cite{dnabook}.
 In this work, we discuss a variant of splicing systems
known as \textit{flat splicing systems}. In flat splicing, if  $x_{0}= x_{0}^{'}u.vx_{0}^{''}$ is spliced
with  $y_{0} = y_{0}^{'}v_{1}y_{0}^{''}$ by
 the rule $r = <u |y_{0}^{'}-y_{0}^{''} |v>$, the string $x_{0}^{'} u. y_{0}^{'} v_{1} y_{0}^{''}.v x_{0}^{''}$ 
is generated. The idea of flat splicing was introduced by Berstel et. al in \cite{berstel}. Also different variants of flat splicing systems 
and their computational power
for linear as well as for circular words 
such as alphabetic flat splicing systems, pure alphabetic flat splicing systems, concatenation systems have been discussed in \cite{berstel}.
The language generating power of the alphabetic flat splicing P systems, its variants
and matrix variant of 
flat splicing systems have been discussed in \cite{kgs-1}, \cite{kgs-3} and  \cite{kgs-2} respectively. Some new model of picture 
generation using the flat splicing rules in arrays has been introduced in \cite{kgs-4}. 

In this paper, we introduce the ideas of two types of derivation languages of labeled flat splicing systems and compared
the families of 
these languages with the family of languages in Chomsky hierarchy. The first one is Szilard languages and the
second one is control languages.
 Szilard languages
are well-known concept in formal language theory. Szilard languages of Chomsky grammars, parallel communicating 
grammar systems, communicating distributive grammar systems etc.  along with their closure properties, decidability aspects, complexity aspects
 have been investigated in  \cite{makinen-2}.
Also, the idea of derivation languages for DNA computing  and membrane computing models have been 
investigated in \cite{prithwi-1} and \cite{prithwi-2} respectively. In \cite{prithwi-1}, derivation languages have been associated with splicing systems and in \cite{prithwi-2}
the derivation languages have been introduced for splicing P systems. Also, the idea of control languages for  spiking neural P systems, transition P systems and
tissue P systems have been investigated in \cite{ajeesh-1,ajeesh-2,ajeesh-3,pan-1}.

In this paper, we show  that the family of Szilard languages of labeled
 flat  finite splicing systems and $REG, CF$ and $CS$ are incomparable.
Moreover,
any non-empty regular, non-empty context-free and recursively enumerable language can be obtained as homomorphic
image of the Szilard language of labeled  flat  finite splicing systems of type $(1, 2), (2, 2)$ 
and $(4, 2)$ respectively. We also introduce the idea of control languages for labeled flat splicing systems and
 show that unlike in the case of Szilard languages, any non-empty regular and context-free language
can be obtained as control language by labeled  flat finite splicing systems of type $(1,2)$ and $(2, 2)$ respectively. Moreover, any recursively enumerable language
can be obtained as control language of labeled  flat finite splicing systems of type $(4,2)$ when the splicing rules can have $\lambda$-label.  

The paper is organized as follows: 
in section \ref{preliminaries} we recall the basic definitions required for this paper. We  introduce Szilard and control languages 
for the labeled flat splicing systems in section \ref{labeled flat splicing system}
and in section \ref{control languages of labeled flat splicing systems} respectively. 
In section \ref{results} and \ref{results1}, we  discuss the main results
of this paper. In section \ref{conclusion}, we discuss conclusion and  open problems which can be investigated further.

\section{Preliminaries}
\label{preliminaries}

For the basic definitions and notions of formal language theory we refer to \cite{handbook}.

\textit{Chomsky  normal form}:
For every context-free grammar $G$, a grammar $G^{'} = (N, T, S, P)$ can be effectively
 constructed where the rules in $P$ are of
 the form $A \rightarrow BC$ and $A \rightarrow a$,  $A, B, C \in N, a \in T$ such that $L(G) \setminus \{\lambda\} = L(G^{'}) \setminus \{\lambda\}.$
 
\textit{Greibach normal form: } A context-free grammar $G = (N, T, S, P)$ is said to be in Greibach normal form if the
rules in $P$ are of the form $A \rightarrow a \alpha, A \in N, a \in T, \alpha \in N^{*}$.
 
 \textit{Type-$0$-grammar}:
A type-$0$-grammar is a construct of the form $G = (N, T, S,$
$ P)$ where $N$ is the non-terminal alphabet and $T$ is the terminal alphabet such that
$N \cap T = \emptyset$. The starting symbol $S \in N$ and the rules in $P$ are ordered
 pairs $(u, v)$ where $u \in (N \cup T)^{*}N(N \cup T)^{*}$
and $v \in (N \cup T)^{*}$.
 
\textit{Kuroda normal form}:
Every type-$0$ grammar $G= (N, T, S, P)$ is in Kuroda normal form  if the rules of the
 grammar $G$ are of the following forms:

$A \rightarrow BC, AB \rightarrow CD, A \rightarrow a, A \rightarrow \lambda$
for $A, B, C, D \in N$ and $a \in T.$

\textit{Homomorphism}:
A homomorphism is a mapping $h$ from $\Sigma^{*}$ to $\Delta^{*}$ where $\Sigma, \Delta$ are alphabets.
 Also, the mapping preserves concatenation,
i.e., $h(v.w) = h(v). h(w), v, w \in \Sigma^{*}.$

\textit{Szilard languages} \cite{handbook}:
Let $G= ( N, T, S, P)$ be Chomsky grammar and $F$ be an alphabet such that the cardinality
 of the set $F$ and $P$ is same. 
Let $f$ be a mapping from $P$ to $F$ such that  for each $p \in P$ a unique label $f(p)$ is 
associated with $p$ and is called the label
of the rule $p$. A derivation in $G$ is called successful if a string over $T$ is generated
 starting from $S$. With each successful derivation
of $G$, a string over $F$ can be associated if the labels of the applied rules of any successful derivation
 are concatenated sequentially. The 
language generated in this manner is called Szilard language of the grammar $G$ and is denoted by $SZ(G).$

\begin{example}
Let $G = (\{S\}, \{a, b\}, S, \{S \rightarrow aSb, S \rightarrow ab\})$ be a context-free grammar.
 The rules are labeled in the
following manner: $f_{1}: S \rightarrow aSb, f_{2}: S \rightarrow ab.$
Hence the Szilard language generated by the grammar $G$ is $SZ(G) = \{ f_{1}^{n}f_{2} ~|~ n \geq 0\}$. 
\end{example}
 The families of regular, context-free, context-sensitive and recursively enumerable languages are
 denoted by $REG, CF, CS$ and $RE$ respectively.

\textbf{\textit{Flat splicing systems}} \cite{berstel}: 
A flat splicing system is a construct of the form $\mathscr{S} = ( A, I, R)$ where $A$ is an alphabet,
 $I$ is a set of words over the alphabet $A$, 
 $R$ is a finite set of splicing rules. The rules in flat splicing system 
are of the form $< \alpha | \gamma - \delta | \beta >$ where $\alpha, \beta, \gamma, \delta$
 are words over the alphabet $A$. The strings $\alpha, \beta, \gamma, \delta$ are called the handles of the rule. 
 If the rule $r =   < \alpha | \gamma - \delta | \beta >$  is applied
 to the pair of strings $(u, v)$ where $u = x \alpha ~.~ \beta y$ and $v = \gamma z \delta$,
 then the string $v$ is inserted in the  $``."$ location of $u$. Hence after application
 of the rule $r$, the string $w = x \alpha ~.~ \gamma z \delta ~.~ \beta y$ is obtained. 
 The location ``." represents the location where the cutting and pasting operation take place.
The flat splicing operation over the two words $u$ and $v$ can be represented as $(u, v) \vdash^{r} w.$
Moreover, a flat splicing system is called finite,  regular, context-free, context-sensitive 
if the corresponding initial set $I$ is finite, regular, context-free and context-sensitive respectively.

The language generated by the flat splicing system  $\mathscr{S} = (A, I, R)$ is denoted by $\mathscr{F(S)}$
which is also the smallest language $L$ 
containing $I$.
It is also closed under $R$, i.e., for $u, v \in L$ and $r \in R, (u, v) \vdash^{r} w \in L.$

In this paper, we introduce the notion of \textit{type} with  
 the rules in flat splicing systems.
A flat splicing system $\mathscr{S} = (A, I, R)$ is called of type $(m,n)$ where
 $m = max \{ |\alpha|, |\beta| | <\alpha | \gamma - \delta | \beta > \in R \}$
and $n = max \{ |\gamma  \delta |  | <\alpha | \gamma - \delta | \beta > \in R\}$
where $u = x \alpha . \beta y, v = \gamma z \delta, ( u, v) = ( x \alpha . \beta y,  \gamma z \delta) \vdash^{r}  x \alpha .   \gamma z \delta.  \beta y, $
and $x, y, \alpha, \beta,  \gamma, \delta, z \in A^{*}$,  $|\gamma| \leq 1, |\delta| \leq 1$ and $| \gamma z \delta | \geq 1 $. The parameters $m$ and
$n$ represent the descriptional complexity measures of flat splicing systems.

Note that whenever a string $\gamma z \delta$ is inserted into another word $x \alpha . \beta y$, the inserted string is represented 
by $\gamma-\delta$, i.e., only by the end letters of the words. Furthermore, if the inserted word (i.e., $|\gamma z \delta| = 1$) is of length $1$, then the word is
denoted as $\epsilon-\gamma_1$ or $\gamma_1-\epsilon$ where $\gamma_1 \in A.$ We have explained these notations with two examples 
in Example \ref{berstel}  and  \ref{letter}.
 Moreover, the language generated by the  flat splicing system $\mathscr{S}$ of 
type $(m, n)$ is denoted as $\mathscr{F}_{n}^{m}(\mathscr{S})$ 
and the families of languages generated by 
 the flat splicing systems of type $(m,n)$ 
is denoted by $FS_{n}^{m}$. When  $m$ and $n$ are not specified, is replaced by $``*"$.
 Also the flat splicing of two strings is simply mentioned as splicing of two strings.
 
 A rule of the flat splicing system $\mathscr{S}$ is called alphabetic if for any rule $r = < \alpha | \gamma-\delta | \beta>$, either
$\alpha, \beta, \gamma, \delta$ are letters or empty strings. A flat splicing system is called alphabetic if 
all the rules present in the system are alphabetic, i.e.,  alphabetic flat splicing systems are flat splicing systems of type $(1,2)$.



\begin{example} \cite{berstel} \label{berstel}
Let  $\mathscr{S} = (A, I, R)$ be a flat splicing system where $A = \{a, b\}, I = \{ab\}$ and $R = \{ <a | a-b | b> \}$.
Since  $ab \in I$, at first  $ab$  splice with $ab$ only  and  generate $a^{2}b^{2}$. This process is continued and the
 strings of the form $a^{n}b^{n}$
are spliced with the strings of the form $a^{n}b^{n}$ only.
Hence
$\mathscr{F}_{2}^{1}(\mathscr{S}) = \{a^{n}b^{n} | n \geq 1\}.$

\end{example}    

In the previous example the strings inserted  must be of length at least two, otherwise the rule cannot be applied. 
Now, we give example of a flat splicing system where strings of length one are inserted in the specified location.

\begin{example}  \label{letter}
Let $\mathscr{S} = (A, I, R)$  where
$A = \{X, Y,a \} , I = \{XY, a\}$ and $R = \{<X | \epsilon-a | Y>\}.$
Hence $\mathscr{F}_{1}^{1}(\mathscr{S}) = \{XaY\}.$
\end{example}

\begin{example}
Let $\mathscr{S} = ( A, I, R)$ be a flat splicing system where $A = \{a, X, b,$
$ Z\},$
$ I = \{ aXb, Z\}, R = \{ <\epsilon | \epsilon - Z | a> \}$.
The application of  $ < \epsilon | \epsilon - Z | a> $ to the words $aXb, Z$ will generate the word $ZaXb$. 
This flat splicing operation is a concatenation operation.  Moreover,  $\mathscr{F}_{1}^{1}(\mathscr{S}) = \{ZaXb\}.$
\end{example}

In this work, we investigate two types of derivation languages of the labeled flat splicing systems. 
 At first, we  introduce the  the idea of Szilard languages of labeled flat
splicing systems.
 Then we  introduce the
idea of control languages of labeled flat splicing systems. 

\section{Labeled flat splicing system}
\label{labeled flat splicing system}
Let $\mathscr{S} = (A, I, R)$ be a flat splicing system. A labeled flat splicing system
 is a construct of the form $\mathscr{LS} = (A, I, R, Lab)$ 
such that $A \cap Lab = \emptyset$. Each rule of the flat splicing system is labeled in one-to-one manner with
 the elements from $Lab$. 
A derivation in flat splicing system is called \textit{terminal derivation} 
if it follows the following pattern:

$(x_{0}, y_{0}) \vdash^{a_{1}} x_{1},
(x_{1}, y_{1}) \vdash^{a_{2}} x_{2},
(x_{2}, y_{2}) \vdash^{a_{3}} x_{3},
\ldots
(x_{n-1}, y_{n-1}) \vdash^{a_{n}} x_{n}$
where $x_{0}, y_{0} , y_{1}, y_{2}, \ldots, y_{n-1} \in I, a_1, a_2, \ldots, a_n \in Lab$. Moreover, no rule from $R$ is 
applicable to the word $x_n \in A^{+}$ 
and the word $a_1a_2 \ldots a_n$ is called a Szilard word. 
The language obtained by concatenating
 the labels of the applied rules in a 
 terminal
derivation is called Szilard language of the labeled flat splicing system $\mathscr{LS}$ and 
$SZ_{n, FAM}^{m}(\mathscr{LS})$ denotes the set of Szilard languages obtained by the labeled flat splicing system $\mathscr{LS}$
of type $(m, n)$ where the initial set $I \in FAM = \{ FIN, REG, CF, CS\}$.
The families of Szilard languages 
 of the labeled flat splicing systems of 
type $(m,n)$, 
 is denoted as $SZLS_{n, FAM}^{m}$. If $m$ and $n$ are not specified, they are replaced by $``*".$



\begin{example}
\label{regularexample}
We give an example of an alphabetic labeled flat finite splicing system which can obtain a
regular language as a Szilard language. 
Let $\mathscr{LS} = (A, I,$
$ R, Lab)$ be a labeled flat splicing system where
$A =  \{ X, A_1,  A^{'}, Y\}$,
$I = \{XA_1Y, A_1, A^{'}\}$,
$R = \{ a: <A_1 | \epsilon - A_1 | Y>; c: < A_1 | \epsilon - A^{'} | Y> \}$.

The $a$-rule is applicable any number of times and it can splice the words $XA_1^{n}Y, n \geq 1$ and $A_1$. But after application of the $c$-rule,
the word  $XA_1^{n}A^{'}Y,$
$ n \geq 1$ is obtained. No rule is applicable to this word. 
Hence
 $SZ_{1, FIN}^{1}(\mathscr{LS}) = \{a^{n}c | n \geq 0\}.$
\end{example}

\section{Results}
\label{results}

In this section, we discuss the results related to the Szilard languages of the labeled flat splicing systems. The language
$\{aa\}$ cannot be obtained as Szilard language of any Chomsky grammar. But we  prove that this language
can be obtained as a Szilard language of labeled  flat finite splicing system. Moreover, we show that some regular, non- regular and non-context free languages
cannot be a Szilard language of any labeled  flat finite splicing system. But any non-empty regular, non-empty context-free and recursively 
enumerable language can be obtained as  homomorphic image of the Szilard language of 
labeled flat finite splicing systems of type $(1, 2), (2, 2)$ and $(4, 2)$ respectively. 

%

\begin{theorem}
$\{aa\} \in SZLS_{1, FIN}^{1}.$
\end{theorem}
\begin{proof}

Let $\mathscr{LS} = (A,  I, R, Lab)$ be a labeled alphabetic flat splicing system  where
$A = \{S, X_{a}, Y\},$ 
$I  = \{SYSY, X_{a}\},$
$R = \{ a:<S |\epsilon -X_{a}| Y>\}$,
$Lab = \{a\}$.
Hence $SZ_{1, FIN}^{1}(\mathscr{LS}) = \{aa\}$.


\end{proof}

In the next two results,  we show that the regular language $\{a^{n} | n \geq 1\}$  cannot be a Szilard 
language of any labeled flat  finite splicing system. But it can be obtained as a Szilard 
language if $I$ is regular.

\begin{theorem} \label{notregular}
$\{a^{n} ~|~ n \geq 1\} \notin SZLS_{*, FIN}^{*}.$ 

\end{theorem}

\begin{proof}
Let us assume that there exists a labeled flat finite  splicing system $\mathscr{LS} = ( A, I, R, Lab)$ such that 
$SZ_{n, FIN}^{m} (\mathscr{LS}) = \{a^{n} | n \geq 1\}$ where $ R = \{a : <u | u_{1} - v_{1} | v> \}$ and $x_i^{1} uv x_i^{2}, u_1 \delta v_1 \in I,
x_i^{1}, x_i^{2}, u, v, u_1, v_1, \delta \in A^{*}, |u_1| \leq 1, |v_1| \leq 1$.
 Now since 
 $a \in \{a^{n} | n \geq 1\}$,
 there exists a  derivation 
 
  $(x_{0}^{0}, y_{0}^{0}) \vdash^{a} z_{1}^{0} \hspace{5cm} \ldots (1)$

 where $x_{0}^{0} = x_{0}^{1}u . v x_{0}^{2}, y_{0}^{0} \in I$ and
   $a : <u | u_{1} - v_{1} | v>$ is not applicable to $z_{1}^{0}.$

Similarly, the terminal derivation for $a^{2}$ is as follows:


$ (x_{0}^{1}, y_{0}^{1}) \vdash^{a} z_{1}^{1}$,

$ (z_{1}^{1}, y_{1}^{1}) \vdash^{a} z_{2}^{1} \hspace{5cm} \ldots (2)$ 

where $x_{0}^{1}, y_{0}^{1}, y_{1}^{1} \in I$ and   $a : <u |  u_{1}-v_{1} | v>$ is
 not applicable to $z_{2}^{1}.$


Again, the terminal derivation of $a^{3}$ is as follows:

$ (x_{0}^{2}, y_{0}^{2}) \vdash^{a} z_{1}^{2}$;

$ (z_{1}^{2}, y_{1}^{2}) \vdash^{a} z_{2}^{2}$;

$ (z_{2}^{2}, y_{2}^{2}) \vdash^{a} z_{3}^{2}  \hspace{5cm}  \ldots (3)$

where $x_{0}^{2}, y_{0}^{2}, y_{1}^{2}, y_2^{2} \in I$ and   $a : <u |  u_{1}-v_{1} | v>$ is
 not applicable to $z_{3}^{2}.$

Since, the labeled flat splicing system $\mathscr{LS}$ contains only 
one $a$-rule, then from the above derivations it is clear that 
either after application of $a$-rule, the subword $uv$ is again obtained 
or the pair of words $(x_0^{i}, y_0^{i})$ are distinct (i.e., at least one of term of the pair is different 
from any other initial pair of words). 

If  after application of the $a$-rule again a subword $uv$ is
obtained, then no terminal derivation is obtained. Hence $\{a^{n} | n \geq 1\}$ cannot 
be a Szilard language of the labeled flat  finite splicing system $\mathscr{LS}$. 

In the second case, to obtain  $\{ a^{n} | n \geq 1\}$ as Szilard language, the 
pairs $(x_0^{i}, y_0^{i})$ must be distinct, where $x_0^{i}, y_0^{i} \in I$. But $I$ is finite. Hence $\{ a^{n} | n \geq 1\}$
cannot be a Szilard language of a labeled  flat  finite splicing systems.




\end{proof}

Now we show that $\{a^{n} | n \geq 1\}$ can be  a Szilard language of labeled flat regular splicing system.
\begin{theorem}

$\{a^{n} | n \geq 1\}$ can be obtained as  Szilard language of a labeled   flat alphabetic regular splicing system, i.e., 
$\{a^{n} | n \geq 1\} \in SZLS_{1, REG}^{1}$.
\end{theorem}

\begin{proof}

Let $\mathscr{LS} = (A, I, R, Lab)$ be a labeled flat splicing system where
 $A = \{X, A_1, A^{'}, Y\}$,
 $I = \{XA_1^{n}Y | n \geq 2\} \cup \{A^{'}\}$,
 $R = \{
a: <A_1 | \epsilon-A^{'} | A_1> \}$, $Lab = \{a\}$.
%
%
%
Hence  $SZ_{1, REG}^{1}(\mathscr{LS}) = \{a^{n} | n \geq 1\}$.

\end{proof}

\begin{theorem} \label{finiteregular}
$SZLS_{1, FIN}^{1} \cap (REG \setminus FIN) \neq \emptyset$.
\end{theorem} 
\begin{proof}
Follows from Example \ref{regularexample}.
\end{proof}

\begin{theorem} \label{finitecontextfree}
$SZLS_{1, FIN}^{2} \cap (CF \setminus REG) \neq \emptyset$.
\end{theorem}
\begin{proof}

 Let $\mathscr{LS} = (A, I, R, Lab)$ be a labeled flat splicing system where $A = \{ X, A_1, Y, A_2\}, I = \{ XA_1Y, A_1, A_2\},
R = \{ a:< A_1 | \epsilon-A_1 | Y>,  b: < A_1 | \epsilon- A_2 | Y>, c: < A_1 | \epsilon-A_2 | A_1A_2> \}$.
Hence $SZ_{1, FIN}^{2}(\mathscr{LS}) = \{ a^n b c^{n+1} | n \geq 1\}$. 
\end{proof}

%
%

\begin{theorem} \label{finitecontextsensitive}
$(CS \setminus CF) \cap SZLS_{1, FIN}^{2} \neq \emptyset.$
\end{theorem}
\begin{proof}
Let $\mathscr{LS} = (A,  I, R, Lab)$ be a labeled flat splicing system where
$A = \{X,A, A^{'}, A^{''}, Y\}$,
$I = \{XAAY, A, A^{'}, A^{''}\}$,
$R = \{ a: <A|\epsilon -A| A>, b: <A|\epsilon - A^{'}| A>, c: <A^{'}A|\epsilon -A^{''}| \epsilon> \}$,
$Lab = \{ a, b, c\}. $

On application of the $a$- rule, one 
$A$ is added between two $A$'s of the word $XAAY$.  Similarly, application of the $b$-rule inserts  $A^{'}$
between the two $A$'s 
of the subword $AA$ and generates the subword $AA^{'}A$. Similarly, when $c$-rule is applied,
 the 
string $AA^{'}$ is replaced by $AA^{''}A^{'}$. Hence terminal derivation can be obtained in $\mathscr{LS}$
if $a$, $b$ and $c$-rules are applied same number of time in order and 
$SZ_{1, FIN}^{2}(\mathscr{LS}) \cap a^{*}b^{*}c^{*} = \{ a^{n} b^{n+1} c^{n+1} ~|~ n \geq 1\}.$ 
Since $\{a^{n} b^{n+1} c^{n+1} | n \geq 1\}$ 
is a context-sensitive language and context-sensitive languages are closed under intersection with regular languages.
 The language $SZ_{1, FIN}^{2}(\mathscr{LS})$
must be 
context-sensitive.
\end{proof}

Now we show that the non-regular context-free language $\{ a^{n} b^{n} | n \geq 1\}$ can be obtained as a Szilard language of 
labeled flat finite splicing system.

\begin{theorem} \label{contextfreenoszilard}
$\{ a^{n} b^{n} | n \geq 1\} \notin SZLS_{*, FIN}^{*}.$ 
\end{theorem}
\begin{proof} 
Let $\mathscr{LS} = ( A, I, R, Lab)$ be a labeled flat finite splicing system such that $SZ_n^{m}(\mathscr{LS}) = \{ a^{n} b^{n} | n \geq 1\}$
where $a: < u | u_1 - v_1 | v>, b: < u_2 | u_3 - v_3 | v_2>, u_1 \delta_1 v_1, u_3 \delta_3 v_3 \in I$.

Since, $ab \in SZ_n^{m}(\mathscr{LS})$, there exists a terminal derivation 

$(x_0^{0}, y_0^{0}) \vdash^{a} z_1^{0}$

$(z_1^{0}, y_1^{0}) \vdash^{b} z_2^{0}$ 

where no rules of $R$ is applicable 
to the word $z_2^{0}$.

Similarly, $a^2b^2 \in SZ_n^{m}(\mathscr{LS})$. Hence there exists a derivation 

$(x_0^{1}, y_0^{1}) \vdash^{a} z_1^{1}$

$(z_1^{1}, y_1^{1}) \vdash^{a} z_2^{1}$

$(z_2^{1}, y_2^{1}) \vdash^{b} z_3^{1}$

$(z_3^{1}, y_3^{1}) \vdash^{b} z_4^{1}$

where $x_0^{1}, y_0^{1}, y_1^{1}, y_2^{1}, y_3^{1} \in I$ and no rule is applicable to $z_4^{1}$.

Also, $a^{3}b^{3} \in SZ_n^{m}(\mathscr{LS})$. Hence there exists a derivation

$(x_0^{2}, y_0^{2}) \vdash^{a} z_1^{2}$

$(z_1^{2}, y_1^{2}) \vdash^{a} z_2^{2}$

$(z_2^{2}, y_2^{2}) \vdash^{a} z_3^{2}$

$(z_3^{2}, y_3^{2}) \vdash^{b} z_4^{2}$

$(z_4^{2}, y_4^{2}) \vdash^{b} z_5^{2}$

$(z_5^{2}, y_5^{2}) \vdash^{b} z_6^{2}$

where $x_0^{1}, y_0^{1}, y_1^{1}, y_2^{1}, y_3^{1} \in I$ and no rule is applicable to $z_6^{2}$.

 $R$ contains two rules $a: < u | u_1 - v_1 | v>$ and $b: < u_2 | u_3-v_3 | v_2>$ where 
$u, v, u_2, v_2 \in A^{*}$ and $u_1, v_1, u_3, v_3 \in A$.  
Proceeding in the similar manner as in Theorem \ref{notregular},
we can say that to obtain $\{ a^{n} b^{n} | n \geq 1\}$ as a Szilard language of labeled flat
finite splicing system $\mathscr{LS}$, either after application of  $a$-rule and $b$-rule, the subwords $uv$ and $u_2v_2$ are obtained
again in the newly generated 
words or  at least one of the words in the initial pair $(x_0^{i}, y_0^{i})$ is distinct from others pairs in terminal derivations.


If both $a$-rule and $b$-rule obtain the subwords $uv$ and $u_2v_2$ in the newly generated words, then
no terminal  derivation is obtained. 
Also if at least one word is distinct in the initial pair of words $(x_i^{0}, y_i^{0})$ then $I$ cannot be finite - a contradiction.
 
Hence  $\{ a^{n} b^{n} | n \geq 1\}$ cannot be a Szilard language of any labeled flat finite splicing system.

\end{proof}

\begin{theorem}
$\{ a^{n} b ^{n} | n \geq 1\} \in SZLS_{ 1, REG}^{3}.$
\end{theorem}
\begin{proof}
Let $\mathscr{LS} = ( A, I, R, Lab)$ be a labeled flat splicing system where $A = \{ X, A_1, B_1, Y\}, I = \{ XA_1^{n}B_1Y | n \geq 2\} \cup \{ B_1, X\}, 
R = \{ a: < A_1 | \epsilon-B_1 | A_1B_1>, b: < XA_1B_1 | \epsilon - X | A_1B_1> \}$.
Hence $SZ_{1, REG}^{3}(\mathscr{LS}) = \{ a^{n} b^{n} | n \geq 1\}.$ 
\end{proof}

%

\begin{theorem} \label{contextsensitivenoszilard}
$\{ a^{n} b^{n} c^{n} | n \geq 1\} \notin SZLS_{*, FIN}^{*}.$ 
\end{theorem}
\begin{proof}
Follows from the proof of Theorem \ref{contextfreenoszilard}.
\end{proof}

\begin{theorem}
$\{ a^{n} b ^{n} c^{n} | n \geq 1\} \in SZLS_{1, REG}^{4}.$
\end{theorem}

\begin{proof}
Let $\mathscr{LS} = ( A, I, R, Lab)$ be a labeled flat splicing system where
$A = \{X, B_1, C_1, Y\}, I = \{XB_1^{n}C_1Y, C_1, X_1, Y_1\}, 
R = \{ a: < B_1| \epsilon-C_1 | B_1C_1>, b: <XB_1C_1 | \epsilon - X | B_1C_1>, 
c: <B_1C_1 | \epsilon- Y | XB_1C_1Y>\}.$
Hence $SZ_{1, REG}^{4}(\mathscr{LS}) = \{a^{n} b^{n} c^{n}  | n \geq 1\}$.
\end{proof}

\begin{theorem}
The families $SZLS_{n, FIN}^{m}$ and $REG$ (resp. $CF, CS$ ) are incomparable.
\end{theorem}
\begin{proof}
Follows from the Theorem \ref{notregular}, \ref{finiteregular}, \ref{finitecontextfree}, \ref{finitecontextsensitive}, \ref{contextfreenoszilard} 
and \ref{contextsensitivenoszilard}.
\end{proof}

\textbf{Open problem: Can labeled flat  regular splicing systems  obtain any recursively enumerable language as Szilard 
language?}


\subsection{Decision problems}

For any regular language $R$ and labeled flat splicing system $\mathscr{LS}$ we can prove the following decision problems.

\begin{theorem} \label{decidable1}
$R \subseteq SZ_{n, FIN}^{m}(\mathscr{LS})$ is decidable. 
\end{theorem}

\begin{proof}
Let $\mathscr{LS} = ( A, I, R, Lab)$ be a labeled flat finite splicing system of type $( m, n)$. Let $x = x_0x_1 \ldots x_n \in R$. If starting from a word $w \in I$, if the flat splicing rules
with label $x_0, x_1, \ldots, x_n$ are applied in order, a string $w_1 \in A^{+}$ is obtained. If no rule of $\mathscr{LS}$
is applicable to $w_1$, then the word $x_0 x_1 \ldots x_n \in SZ_{n, FIN}^{m}(\mathscr{LS})$. Otherwise, $x \notin SZ_{n, FIN}^{m}(\mathscr{LS})$. 
Hence it is decidable whether a regular language $R$ is contained inside the Szilard language of a flat splicing system
of type $(m, n)$ with finite initial set.
\end{proof}

\begin{theorem}
$SZ_{n, FIN}^{m}(\mathscr{LS}) \subseteq R$ is decidable.
\end{theorem}

\begin{proof}
Since $R$ is a regular language, there exists a deterministic finite automata $D$ accepting  it.
Moreover, the problem $M = \{ < D, x> | D$ accepts the string $x\}$ is decidable. Let $x_1 \in SZ_{n, FIN}^{m}(\mathscr{LS})$, then 
it is also decidable whether or not $D$ accepts the string $x_1$.  Hence the problem $SZ_{n, FIN}^{m}(\mathscr{LS}) \subseteq R$
is also decidable.

\end{proof}

\textbf{Open problem:
Let $\mathscr{LS}$ be a labeled flat finite splicing system and 
$R$ be a regular language, then  
$SZ_{n, FIN}^{m}(\mathscr{LS}) = R$ is decidable or not.}


 We have already proved that not all regular languages can be Szilard language of labeled flat finite splicing
 systems. But in the next result we show that any non-empty regular language can be obtained as homomorphic image of the
Szilard language of the labeled  flat  finite splicing systems of type $(1, 2)$.

\begin{theorem}
Any  non-empty regular language can be obtained as a homomorphic image of the Szilard language of the  labeled  flat finite splicing system of type $(1, 2)$.
\end{theorem}
\begin{proof}
Let $L$ be a $\lambda$-free regular language.  There exist a grammar $G = ( N, T, S, P)$ which generates $L$, i.e., $L = L(G)$. 
We construct a labeled  flat splicing system which simulates the rules in the grammar $G$.
Initially, the non-terminals $N$ of $G$ are rewritten using the symbols $D_{i}, 1 \leq i \leq n$, 
 starting from $D_{1} = S$ and the labeled  flat splicing system $\mathscr{LS}$ is constructed in such a manner  that $L = L(G) = h(SZ_{2, FIN}^{1}(\mathscr{LS}))$
 where $h$ is a homomorphism and $SZ_{2, FIN}^{1}(\mathscr{LS})$ denotes the Szilard language of the labeled flat splicing system of type $(1, 2)$.
 Now the rules in $P$ are of the form $D_{i} \rightarrow a D_{i}$, $D_{i} \rightarrow a D_{j} ( i \neq j )$,
 and $D_{i} \rightarrow a$, where $D_{i}, D_{j} \in N$, and $a \in T$. 
 
\noindent  Let $\mathscr{LS} = (A,  I, R, Lab)$ be a labeled flat splicing system where:
  \begin{description} 
\item[$\bullet$]  $A = \{ X, Y, D_{1}, D_{2}, \ldots, D_{n}\} \cup \{Y_{a} | a \in T\};$ 

\item[$\bullet$]$I = \{ XD_{1}Y\} \cup \{Y_{a} D_{i} | D_{i} \rightarrow a D_{i} \in P\} \cup \{Y_{a} D_{j} | D_{i} \rightarrow a D_{j} \in P\} 
\cup \{ Y_{a} | D_{i} \rightarrow a \in P\}$; 

 \item[$\bullet$]The rules in $R$  are of the form: \\

$a_{D_{i}}^{i}: <D_{i} | Y_{a}- D_{i} | Y>$ ($Y_a-D_i = Y_aD_i$) for $D_{i} \rightarrow a D_{i}, D_{i} \in N, a \in T;$

$a_{D_{j}}^{i}: <D_{i} | Y_{a}- D_{j} | Y>$ ($Y_a-D_j = Y_aD_j$)  for $D_{i} \rightarrow a D_{j} (i \neq j), D_{i}, D_{j} \in N, a \in T;$

$a^{i}: <D_{i}|\epsilon - Y_{a}| Y>$ for $D_{i} \rightarrow a, a \in T$;
\item[$\bullet$] $Lab = \{a_{D_{i}}^{i} ~|~  D_{i} \rightarrow a D_{i}, D_{i} \in N, a \in T\} \cup \{a_{D_{j}}^{i} ~|~  D_{i} \rightarrow a D_{j}, D_{i}, D_{j} \in N, a \in T\} \cup \{a^{i} ~|~ D_{i} \rightarrow a, a \in T\}$.
 \end{description}
 
 The non-erasing homomorphism $h: (Lab)^{*} \rightarrow T^{*}$ is defined in the following manner:
 $h(a_{D_{i}}^{i}) = h(a_{D_{j}}^{i}) = a$ and $h(a^{i}) = a$ where $a_{D_{i}}^{i},  a_{D_{j}}^{i}, a^{i} \in Lab.$
 
 The rule $D_{i} \rightarrow aD_{i}$  in $G$  is simulated by the labeled splicing rule
 $a_{D_{i}}^{i}: <D_{i} | Y_{a}- D_{i} | Y>$  and  the rule $D_{i} \rightarrow aD_{j}$ is simulated by the
splicing rule $a_{D_{j}}^{i}: <D_{i} | Y_{a}- D_{j} | Y>$. 
 Also, the terminating rule $D_{i} \rightarrow a$ is simulated by $a^{i}: <D_{i}|\epsilon - Y_{a}| Y>$.
 Every $w\in L(G)$ can be generated after application of the rules in $G$ and there exist labeled
splicing rules simulating the rules in $G$. The application of the labeled splicing 
	rules starting from $XD_{1}Y$ in the same sequence as in the derivation of $w \in T^{*}$, 
	generates a string over $A_{1}$  such that no splicing rule is applicable to it. Also,
	concatenation of the labels of the splicing rules generate $ w_{1} \in (Lab)^{*}$. Since,  
	$h(a_{D_{i}}^{i}) = h(a_{D_{j}}^{i}) = a$ and $h(a^{i}) = a$, the homomorphic image of the string $w_{1}$
	is $w$, i.e., $h(w_{1}) = w$.
   Hence $w \in h(SZ_{2, FIN}^{1}(\mathscr{LS})).$
   This imply, $L= L(G)  \subseteq SZ_{2, FIN}^{1}(\mathscr{LS}).$

   It only remains to prove the inclusion $h(SZ_{2, FIN}^{1}(\mathscr{LS})) \subseteq L(G)$.
	So, let $w \in h(SZ_{2, FIN}^{1}(\mathscr{LS}).$ 
   Hence there exists  $w_{1} \in SZ_{2, FIN}^{1}(\mathscr{LS})$ such that  $h(w_{1}) = w$.
   This inclusion can be proved by observing the one-to-one correspondence between the 
   rules in $\mathscr{LS}$ and $G$. 
   
%

\end{proof}

It was proved by P{\u a}un  in \cite{Paun} that some context-free languages  cannot be represented as a homomorphic
image of the Szilard language of any context-free language.

\begin{theorem} \cite{Paun}
The families of context-free languages and homomorphic image of context-free languages are incomparable.
\end{theorem}

It was proved in \cite{Paun} that  $\{a^{n}b^{n} | n \geq 1\}$
cannot be obtained as a homomorphic image of the Szilard 
languages of the context-free languages.
    But we prove that  any context-free language can be obtained as a homomorphic image of
Szilard language of the labeled finite  flat splicing systems of type $(2, 2)$.

\begin{theorem}
Any context-free language can be obtained as a homomorphic image of the Szilard language of the
 labeled  flat finite splicing system of type $(2,2)$.
\end{theorem}
\begin{proof}

Let $L $ be a non-empty context-free language and let $G=(N,T, S, P)$ be a grammar such that $L = L(G).$ 
Let the grammar $G$ is in Chomsky normal form and  the rules in $P$ are of the form, 
$A_1 \rightarrow B_1C_1$ and $A_1 \rightarrow a$, where $A_1, B_1, C_1 \in N, a \in T$. Each element of the language $L$
can be generated by initial application of the recursive rules and then by application of the terminal rules 
in the left-most manner. Also, each  rule of $R$ is associated with unique label $r_{i}$.

We construct a labeled flat splicing systems
$\mathscr{LS} = ( A, I, R, Lab)$ 
such that $L = h (SZ_{2, FIN}^{2}( \mathscr{LS}))$ where $h$ is a morphism from $(Lab)^{*}$ to $T^{*}$ and
 $SZ_{2, FIN}^{2}(\mathscr{LS})$ denotes the Szilard language of the labeled 
 flat splicing system
 $\mathscr{LS}$ of type $(2,2)$.

 Let  $\mathscr{LS} = ( A,  I, R, Lab)$ be a labeled  flat splicing system  
 where: 
\begin{itemize}
\item[$\bullet$] $A = \{X, Y, E \} \cup N  \cup  \Delta_{1}  \cup \Delta_{2}   \cup \{[r_{k1}], [r_{m}]\}$ where $\Delta_{1} = \{ [r_{i}] ~|~ r_{i}: A_1 \rightarrow B_1C_1\}, \Delta_{2} = \{ [r_{i}] ~|~ r_{i}: A_1 \rightarrow a\};$
\item[$\bullet$] $I = \{ XSEY \} \cup \{[r_{i}]B_1C_1 | r_{i}: A_1 \rightarrow B_1C_1 \in P\} \cup \{ [r_{a}] | r_{a}: A_1 \rightarrow a\} \cup \{ [r_{k1}], [r_{m}]\}$;


\item[$\bullet$]
$R$ contains the following rules: 

for $r_{i}: A_1 \rightarrow B_1C_1$:

 $[r_{i}]^{1}: <A_1 | [r_{i}]  -C_1 | \alpha_{1} \alpha_{2}>$  where  $\alpha_{1} \in N \cup \{E\}, \alpha_{2} \in N \cup \{ E, Y\} 
 \cup \{ [r_i] | [r_i] \in \Delta_1\}, \alpha_1 \alpha_2 \notin \{ NY, EN\} \cup \{E [r_i] | [r_i] \in \Delta_1\} ([r_i]-C_1 = [r_i]B_1C_1)$,

$[r_{i}]^{a}: <[r_{m}]A_1 | \epsilon -[r_{a}] | \alpha_{3}>, \alpha_{3} \in N \cup  \{E\}$ for  $r_{i}: A_1 \rightarrow a$.


$[r_{k1}]^{'}: <X |  \epsilon-[r_{m}] |  \alpha_4>, \alpha_4 \in N;$

%
%
%
%

$[r_{k2}]^{'}: <[r_m] \alpha_5 | \epsilon- [r_m] | \alpha_6>, \alpha_5 \in N \cup \Delta_1 \cup \Delta_2, \alpha_6 \in N \cup \Delta_1 \cup \Delta_2 $.

\item[$\bullet$] $Lab = \{ [r_{i}]^{1} ~|~ [r_{i}] \in \Delta_{1}\} \cup \{ [r_{i}]^{a} ~|~ [r_{i}] \in \Delta_{2}\} \cup \{[r_{k1}]^{'}, [r_{k2}]^{'}\}$. 
\end{itemize}

The morphism $h : (Lab)^{*} \rightarrow T^{*}$ is as follows:  $h([r_{i}]^{1})  = h([r_{k1}]^{'}) = h([r_{k2}]^{'}) =  \lambda, h([r_{i}]^{a}) = a$
 where $[r_{i}]^{1}, [r_{k1}]^{'}, [r_{k2}]^{'},  [r_{i}]^{a} \in Lab$ and $a \in T.$
 
We first prove that $L(G) = L \subseteq h(SZ_{2, FIN}^{2}( \mathscr{LS}))$. 
Let $w \in L(G)$ and it can be generated by application of the recursive rules and then by 
left-most application of the terminating rules. 
Each rule of $G$ is associated with a unique label
  and the rule $r_{i}: A_1 \rightarrow B_1C_1$ is simulated by the splicing rule $<A_1 | [r_{i}] - C_1 | \alpha_{1} \alpha_{2}>,$ 
  $\alpha_1 \in N \cup \{E\}, \alpha_2 \in N \cup \{ Y, E\}, \alpha_1 \alpha_2 \notin \{ EN, NY, EE\}$. 
In fact, the string $Xw_{1}A_1w_{2}Y$  where $w_{1}, w_{2} \in A^{*}$ is spliced with the string $[r_{i}]B_1C_1$ and generate the string $Xw_{1}A_1[r_{i}]B_1C_1w_{2}Y$.
Each $r_{i}: A_1 \rightarrow a$ is applied in the left-most manner and is simulated when the string $Xw_{1}[r_{m}]A_1w_{2}Y$ is spliced with $[r_{a}]$ and generate the string $Xw_{1}[r_{m}]A$
$[r_{a}]w_{2}Y$.
The symbol $[r_{m}]$ identifies the left-most non-terminal where the terminal rule can be applied. The $[r_{k1}]^{'}$-rule 
 splice the string $X \alpha w_{3}Y, w_{3} \in A_{1}^{*}, \alpha \in N$ with $[r_{k1}]$ and generate
the string $X[r_{k1}] \alpha w_{3}Y$. The $[r_{k2}]^{'}$ labeled rule  insert  $[r_{m}]$ in a specified location.  
The $[r_{k1}]^{'}$ and $[r_{k2}]^{'}$ labeled rules are constructed in such a manner that they help to identify
the leftmost non-terminal where the terminating rules can be applied.
  If the corresponding 
	labeled splicing rules are applied in $\mathscr{LS}$ in the same sequence
	as in the derivation of $G$ generating $w$,
   a terminal derivation can be obtained in $\mathscr{LS}$, i.e., a string over $A$ is obtained where no rules
	 can be applied.  Also, the
   concatenation of the labels of the applied splicing rules  generate a string over $Lab$, say, $w_{1}$.
		The morphism $h$ replaces  each occurrence of $[r_{i}]^{1}, [r_{k1}]^{'}$ and $[r_{k2}]^{'}$ in $w_{1}$ by the empty string
			and  $[r_{i}]^{a}$  by $a$. Hence each $w \in L(G)$,
     can be represented as $w = h(w_{1}) \in h(SZ_{2, FIN}^{2}(\mathscr{LS}))$ where $w_{1} \in SZ_{2, FIN}^{2}(\mathscr{LS}).$ 

Now to prove the other inclusion $h(SZ_{2, FIN}^{2}(\mathscr{LS})) \subseteq L(G) = L$, let $w \in h(SZ_{2, FIN}^{2}(\mathscr{LS})$.
Hence $w = h(w_{1})$
where $w_{1} \in SZ_{2, FIN}^{2}(\mathscr{LS})$.
 Since $w_{1} \in SZ_{2, FIN}^{2}(\mathscr{LS})$, i.e., concatenation of the labels of the applied splicing rules of a terminal 
derivation in $\mathscr{LS}$ generates $w_{1}$.
	The application of the rules in $G$, starting from $S$ in the same sequence as
	in the terminal derivation of $\mathscr{LS}$ generates $w \in L(G)$. 
		Hence $h(SZ_{2, FIN}^{2}(\mathscr{LS})) \subseteq L(G).$

\end{proof}

In the next result, we show that any recursively enumerable language can be 
characterized by the homomorphic image of the Szilard language of 
the labeled  flat finite splicing system of type $(4, 2)$.

\begin{theorem} \label{rehomomorphism}
Each recursively enumerable language can be obtained as  homomorphic image of the Szilard language
of  labeled finite flat splicing systems of type $(4, 2).$
\end{theorem}
\begin{proof}
Let $L \in RE$ and we know that any $RE$ language can be generated by a grammar $G = (N, T, S, P)$ in Kuroda normal form, i.e., $L(G) = L$. 
The rules of the grammar $G$ are of the 
form $A_1 \rightarrow B_1C_1, A_1B_1 \rightarrow C_1D_1, A_1 \rightarrow a, A_1 \rightarrow \lambda$ where $A_1, B_1, C_1, D_1 \in N, a \in T$. 
Any element $x \in L$ can be
 generated  by application of the recursive rules and then by left-most application of the 
terminating rules \cite{reproof}. 
 In this proof, the labeled flat splicing system $\mathscr{LS} = (A, I, R, Lab)$ is constructed in 
 such  a manner that $L = h(SZ_{2, FIN}^{4}(\mathscr{LS}))$ where $A \cap Lab = \emptyset.$ 
Also, the labeled splicing rules in $\mathscr{LS}$ are constructed by simulating the rules in $P$ where each rule of $P$ is associated 
with a unique label $r_{i} (1 \leq i \leq n)$
if $|P| = n$.

The set $\Delta$ contains the labels of the rules in $P$ and is divided into four parts such that
 $\Delta = \Delta_{1} \cup \Delta_{2} \cup \Delta_{3} \cup \Delta_{4}$, where

$\Delta_{1} = \{ r_{i} ~|~ r_{i}: A_1 \rightarrow B_1C_1 \in P\};$

$\Delta_{2} = \{ r_{i}  ~|~ r_{i}: A_1B_1 \rightarrow C_1D_1 \in P\};$

$\Delta_{3} = \{ r_{i} ~|~ r_{i}: A_1 \rightarrow a \in P\};$

$\Delta_{4} = \{ r_{i} ~|~ r_{i}: A _1 \rightarrow \lambda \in P\};$


\noindent Let $\mathscr{LS} = (A,  I, R, Lab)$ be a labeled flat splicing system, where:
\begin{enumerate}

\item[$\bullet$] $A =  \{X, Y\} \cup N \cup \{ [r_{i}]  ~|~  r_{i} \in \Delta\} \cup \{[r_{m}]\} \cup  \{ k_{a} | r_{i}: A_1 \rightarrow a\}
\cup \{k_{\lambda}, [r_{m}]\} ;$ 
\item[$\bullet$] $I = \{XSY\} \cup \{[r_{i}]B_1C_1 | r_{i}: A_1 \rightarrow B_1C_1 \in P\} \cup \{ [r_{i}]C_1D_1 | r_{i}: A_1B_1 \rightarrow C_1D_1 \in P\} \cup 
\{[r_{i}] | r_{i}: A_1B_1 \rightarrow C_1D_1\} \cup \{ k_{a}^{i} | r_{i}: A_1 \rightarrow a \in P\}
\cup \{k_{\lambda}^{i} | r_{i}: A_1 \rightarrow \lambda \in P\} \cup \{ [r_{m}]\} ;$

\item[$\bullet$] $R$ contains the following rules: 

%
%
%
%
%
%
%
%
%
%

$\mathbf{(R_{11})}$ 
For $r_{i}: A_1 \rightarrow B_1C_1$:

$r_{i}^{1}: <A_1 |   [r_{i}] - C_1 | Y>$, 

$r_{i}^{2}: <A_1 |   [r_{i}] - C_1 | \alpha_{1} Y>, \alpha_{1} \in N$,

$r_{i}^{3}: <A_1 |   [r_{i}] - C_1 | \alpha_{1} \alpha_{2} Y>, \alpha_{1}, \alpha_{2} \in N$,

$r_{i}^{4}: <A_1 |  [r_{i}] - C_1 | \alpha_{1} \alpha_{2} \alpha_{3} Y>,  \alpha_{1}, \alpha_{2}, \alpha_{3} \in N$,

$r_{i}^{5}: <A_1 |  [r_{i}] - C_1 | \alpha_{1} \alpha_{2} \alpha_{3} \alpha_{4}>,$

where $[r_i]-C_1 = [r_i]B_1C_1$,
$\alpha_{1} \in N, \alpha_{2} \in N \cup \Delta_{1}, \alpha_{3} \in N \cup \Delta_{1} \cup \Delta_{2}, \alpha_{4} \in N \cup \Delta_{1} \cup \Delta_{2},$
$\alpha_{2} \alpha_{3} \notin (\Delta_{1})(\Delta_{1} \cup \Delta_{2}),\alpha_{3} \alpha_{4} \notin (\Delta_{1} \cup \Delta_{2})(\Delta_{1} \cup \Delta_{2})$,

$r_{i}^{6}: <A_1 |  [r_{i}] - C_1 | \alpha_{1} \alpha_{2} \alpha_{3} \alpha_{4} \alpha_{5}>,$  
where $[r_i]-C_1 = [r_i]B_1C_1$,
$\alpha_{1} \in N, \alpha_{2} \in \Delta_{2}, \alpha_{3} \in N, \alpha_{4} \in \Delta_{1}$ and 
$\alpha_{2} = \alpha_{5}.$

\vspace{0.2cm}

%
%
%
%
%
%
%
%
%
%
%

$\mathbf{(R_{12})}$ 
For $r_{i}: A_1B_1 \rightarrow C_1D_1$:


$r_{i}^{7}: <A_1B_1 | [r_{i}] - D_1 | \alpha_{1} \alpha_{2}>, \alpha_{1} \in N,  \alpha_{2} \in N,  $


$r_{i}^{8}: <A_1B_1 | [r_{i}]-D_1 | Y>$,

$r_{i}^{9}: <A_1B_1 |  [r_{i}] -D_1 | \alpha_{1} Y>, \alpha_{1} \in N$

where $[r_i] C_1D_1 = [r_i] - D_1$,

The application of the rules $r_{i}^{10}, r_{i}^{11}$ and $r_{i}^{12}$ in order also can simulate the rule $r_{i}: A_1B_1 \rightarrow C_1D_1:$


$r_{i}^{10}: <A_1 | \epsilon - [r_{i}] |  \alpha_{1} \alpha_{2} >,  \alpha_{1} \in N, \alpha_{2} \in \Delta_{1}$,


$r_{i}^{11}: <[r_{i}] \alpha_{1} \beta_{1} | \epsilon - [r_{i}] |  \alpha_{2} \alpha_{3}>, \alpha_{1} \in N, \alpha_{2} \in N,  \alpha_{3} \in \Delta_{1} \cup N, [r_{i}] \in \Delta_{2}, \beta_{1} \in \Delta_{1},$


$r_{i}^{12}: <\beta_{1}  [r_{i}] B_1 |  [r_{i}] - D_1 | \alpha_{1} \alpha_{2}>, \alpha_{1} \in N,  \alpha_{2} \in N \cup \{Y\} \cup \Delta_{1}, [r_{i}] \in \Delta_{2}, \beta_{1} \in \Delta_{1}, [r_i]C_1D_1 = [r_i]- D_1$

and



$r_{i}^{13}: <A_1B_1 |  [r_{i}] - D_1 | \alpha_{1} \alpha_{2} \alpha_{3} \alpha_{4} \alpha_{5}> $,

where

$[r_i] C_1D_1 = [r_i] - D_1$, 
$\alpha_{1} \in N, \alpha_{2} \in \Delta_{2}, \alpha_{3} \in N, \alpha_{4} \in \Delta_{1}$ and $\alpha_{2} = \alpha_{5}.$ 


\vspace{0.2cm}

%
%
%
%
%
%
%
%
%
%
%
%

$\mathbf{(R_{13})}$ 
For $r_{i}: A_1 \rightarrow a$:


$a_{i}^{1}: <XA_1 | \epsilon - k_{a}^{i} |  Y>$,


$a_{i}^{2}: <[r_{m}]A_1 | \epsilon - k_{a}^{i} | Y>$,


$a_{i}^{3}: <[r_{m}]A_1 | \epsilon - k_{a}^{i} | \alpha_{1} Y>, \alpha_{1} \in N$,

$a_{i}^{4}: <[r_{m}]A_1 |  \epsilon - k_{a}^{i} | \alpha_{1} \alpha_{2} Y>, \alpha_{1}, \alpha_{2} \in N$,

$a_{i}^{5}: < [r_{m}]A_1 | \epsilon - k_{a}^{i} |  \alpha_{1} \alpha_{2} \alpha_{3} Y>,  \alpha_{1}, \alpha_{2},  \alpha_{3} \in N$,

$a_{i}^{6}: <[r_{m}]A_1 | \epsilon - k_{a}^{i} | \alpha_{1} \alpha_{2} \alpha_{3} \alpha_{4}>, $

where
$\alpha_{1} \in N, \alpha_{2} \in N \cup \Delta_{1}, \alpha_{3} \in N \cup \Delta_{1} \cup \Delta_{2}, \alpha_{4} \in N \cup \Delta_{1} \cup \Delta_{2},$

$\alpha_{2} \alpha_{3} \notin (\Delta_{1})(\Delta_{1} \cup \Delta_{2}),\alpha_{3} \alpha_{4} \notin (\Delta_{1} \cup \Delta_{2})(\Delta_{1} \cup \Delta_{2})$,

$a_{i}^{7}: < [r_{m}]A_1 | \epsilon - k_{a}^{i} | \alpha_{1} \alpha_{2} \alpha_{3} \alpha_{4} \alpha_{5}>$,  

where
$\alpha_{1} \in N, \alpha_{2} \in \Delta_{2}, \alpha_{3} \in N, \alpha_{4} \in \Delta_{1}$ and $\alpha_{2} = \alpha_{5},$

$r_{m+1}^{i}: <[r_{m}]A_1k_{a}^{i} | \epsilon - [r_{m}] | \alpha_{1} \alpha_{2}>, \alpha_{1} \in N, \alpha_{2} \in \{Y\} \cup N \cup \Delta_{1} \cup \Delta_{2}$.

\vspace{0.2cm}

%
%
%
%
%
%
%
%
%
%
%

$\mathbf{(R_{14})}$ 
For $r_{i}: A \rightarrow \lambda$:

$r_{i}^{14}: <XA_1 |  \epsilon - k_{\lambda}^{i} | Y>$,

$r_{i}^{15}:  <[r_{m}]A_1 |  \epsilon - k_{\lambda}^{i} | Y>$,

$r_{i}^{16}: <[r_{m}]A_1 |  \epsilon - k_{\lambda}^{i} |  \alpha_{1} Y>, \alpha_{1} \in N$,

$r_{i}^{17}: <[r_{m}]A_1 | \epsilon - k_{\lambda}^{i} |  \alpha_{1} \alpha_{2} Y>, \alpha_{1}, \alpha_{2} \in N$,

$r_{i}^{18}: <[r_{m}]A_1 | \epsilon - k_{\lambda}^{i} | \alpha_{1} \alpha_{2} \alpha_{3} Y >, \alpha_{1}, \alpha_{2}, \alpha_{3} \in N$,

$r_{i}^{19}:  <[r_{m}]A_1 | \epsilon - k_{\lambda}^{i} |  \alpha_{1} \alpha_{2} \alpha_{3} \alpha_{4}>, $

where
$\alpha_{1} \in N, \alpha_{2} \in N \cup \Delta_{1}, \alpha_{3} \in N \cup \Delta_{1} \cup \Delta_{2}, \alpha_{4} \in N \cup \Delta_{1} \cup \Delta_{2},$

$\alpha_{2} \alpha_{3} \notin (\Delta_{1})(\Delta_{1} \cup \Delta_{2}),\alpha_{3} \alpha_{4} \notin (\Delta_{1} \cup \Delta_{2})(\Delta_{1} \cup \Delta_{2})$,

$r_{i}^{20}: <[r_{m}]A_1 | \epsilon - k_{\lambda}^{i} | \alpha_{1} \alpha_{2} \alpha_{3} \alpha_{4} \alpha_{5}>$,

where
$\alpha_{1} \in N, \alpha_{2} \in \Delta_{2}, \alpha_{3} \in N, \alpha_{4} \in \Delta_{1}$ and $\alpha_{2} = \alpha_{5}.$

$r_{m+2}^{i}:  <[r_{m}]A_1k_{\lambda}^{i} |  \epsilon - [r_{m}] | \alpha_{1} \alpha_{2}>, \alpha_{1} \in N, \alpha_{2} \in \{Y\} \cup N \cup  \Delta_{1} \cup \Delta_{2}$.

\vspace{0.2cm}

%
%
%
%
%
%
%
%
%

%
%
%
%
%
%
%
%
%
 
 $\mathbf{(R_{15})}$ 
 $r_{m}: <X \alpha_{1} \beta_{1} | \epsilon - [r_{m}] | \alpha_{2}>, \alpha_{1}, \alpha_{2} \in N, \beta_{1} \in \Delta_{1}$,

$r_{m+1}: <[r_{m}] \alpha_{1} \alpha_{2} \beta_{1} | \epsilon - [r_{m}] | \alpha_{3}>$, $\alpha_{1}, \alpha_{2}, \alpha_{3} \in N, \beta_{1} \in \Delta_{2}$,

$r_{m+2}: <[r_{m}] \alpha_{1} \beta_{1} |  \epsilon - [r_{m}] | \alpha_{2} \beta_{2} \beta_{1}>, \alpha_{1}, \alpha_{2} \in N, \beta_{1} \in \Delta_{2}, \beta_{2} \in \Delta_{1}$,

 $r_{m+3}: <[r_{m}] \alpha_{1} \beta_{1} |  \epsilon - [r_{m}] |  \alpha_{2} \alpha_{3}>, \alpha_{1} \in N, \alpha_{2} \in N,  \alpha_{3} \in N  \cup \Delta_{1} \cup \Delta_{2}, \beta_{1} \in \Delta_{1} $,

 $r_{m+4}: <[r_{m}] \alpha_{1} \beta_{1} \beta_{2} | \epsilon - [r_{m}] |  \alpha_{2} \beta_{2}>, \alpha_{1}, \alpha_{2} \in N, \beta_{2} \in \Delta_{2}, \beta_{1} \in \Delta_{1}$,

 $r_{m+5}: <[r_{m}] \alpha_{1} \beta_{1} \beta_{2} | \epsilon - [r_{m}] | \alpha_{2} \alpha_{3} \beta_{2}>,  \alpha_{1} \in N, \alpha_{2} \in N, \alpha_{3} \in \Delta_{1}, $
 $\beta_{1} \in \Delta_{1}, \beta_{2} \in \Delta_{2}$,

 $r_{m+6}: <[r_{m}] \alpha_{1} \beta_{1} |  \epsilon - [r_{m}] | \alpha_{2} \alpha_{3} >, \alpha_{1}, \alpha_{2} \in N, \alpha_{3} \in N \cup \Delta_{1} \cup \Delta_{3} \cup \Delta_{4} , \beta_{1} \in \Delta_{2}$.

\item[$\bullet$] $Lab= \{r_{i}^{1}, r_{i}^{2}, r_{i}^{3}, r_{i}^{4}, r_{i}^{5}, r_{i}^{6} ~|~ [r_{i}] \in \Delta_{1}\}$ 

$\cup \{r_{i}^{7}, r_{i}^{8}, r_{i}^{9}, r_{i}^{10}, r_{i}^{11}, r_{i}^{12}, r_{i}^{13} ~|~ [r_{i}] \in \Delta_{2}\}$

$\cup \{a_{i}^{1}, a_{i}^{2}, a_{i}^{3}, a_{i}^{4}, a_{i}^{5}, a_{i}^{6},  a_i^{7}, r_{m+1}^{i} ~|~ [r_{i}] \in \Delta_{3}\}$

$\cup \{ r_{i}^{14}, r_{i}^{15}, r_{i}^{16}, r_{i}^{17}, r_{i}^{18}, r_{i}^{19}, r_{i}^{20}, r_{m+2}^{i} ~|~ [r_{i}] \in \Delta_{4}\}$

$\cup \{ r_{m},  r_{m+1}, r_{m+2},$
$ r_{m+3},
 r_{m+4}, r_{m+5}, r_{m+6}\}.$

\end{enumerate}

The homomorphism is defined in the following manner:

$h: (Lab)^{*} \rightarrow T^{*}$  by $h(a_{i}^{1}) = h(a_{i}^{2}) = h(a_{i}^{3}) = h(a_{i}^{4}) = h(a_{i}^{5}) = h(a_{i}^{6}) = h(a_{i}^{7}) = a$ and $h(l) = \lambda,$
 for $l \in Lab \setminus \{a_{i}^{1}, a_{i}^{2}, a_{i}^{3}, a_{i}^{4}, a_{i}^{5}, a_i^{6}, a_i^{7}\}$.
 
 Her we discuss summary of the proof of Theorem \ref{rehomomorphism}.
The detailed simulation of the rules has been discussed in the appendix.


Now in order to prove $L = h(SZ_{2, FIN}^{4}(\mathscr{LS}))$, we have to prove the inclusions $L \subseteq h(SZ_{2, FIN}^{4}(\mathscr{LS}))$
and $h(SZ_{2, FIN}^{4}(\mathscr{LS})) \subseteq L$, i.e., we have to prove that any element of the recursively enumerable language $L$
can be obtained as homomorphic image of a Szilard language of a labeled finite flat splicing system. Moreover, no other word can 
be obtained as homomorphic image of the Szilard language of the labeled finite flat splicing system $\mathscr{LS}$ except the
words from $L$. 

To prove the inclusion $ L = L(G) \subseteq h(SZ_{2, FIN}^{4}(\mathscr{LS}))$, 
let us assume that  $w \in L = L(G).$
Now, $w \in L$ can be generated at first by application of the recursive rules and then by left-most application 
of the terminal rules.  Starting from $XSY$ if the rules in labeled flat splicing system $\mathscr{LS}$  are applied to 
  in the same sequence,
a word over $A$ is generated and  no further computation will be possible. The concatenation of the labels of the
 applied rules  generate  a word  $w^{'} \in (Lab)^{*}$. Since,  the morphism $h$ only maps the labels
  $ a_{i}^{1}, a_{i}^{2}, a_{i}^{3}, a_{i}^{4}, a_{i}^{5}, a_i^{6}$ and $a_i^{7}$ to $a$ and others
 to $\lambda$, the string $w$ can be represented as $w = h(w^{'}).$ 

   Hence $L = L(G) \subseteq h(SZ_{2, FIN}^{4}(\mathscr{LS})).$
 
To prove the second inclusion  $h(SZ_{2, FIN}^{4}(\mathscr{LS})) \subseteq L(G) = L,$ 
 let $x \in h(SZ_{2, FIN}^{4}(\mathscr{LS}))$. Hence there exists a $x_{1} \in SZ_{2, FIN}^{4}(\mathscr{LS})$
such that $x = h(x_{1}).$
 The word $x_{1}$ is obtained by concatenating the labels of a terminal derivation of $\mathscr{LS}$.
 Since, the flat splicing rules are simulated from the rules in $G$ and no extra derivation is possible,  the application of the
  rules in $G$ in the same sequence 
 generates  $x \in T^{*}$. Hence $h(SZ_{2, FIN}^{4}(\mathscr{LS})) \subseteq L(G) = L.$
This will imply, $h(SZ_{2, FIN}^{4}(\mathscr{LS})) = L(G) = L.$
\end{proof}

In the next section, we associate the idea of control languages with 
 labeled flat  finite splicing systems. Unlike in the case 
of Szilard languages, the same label can be assigned to multiple 
 rules  but one rule
cannot have multiple labels. We show that although there exists
 regular languages which cannot be  Szilard
language by any labeled  flat finite  splicing systems, any non-empty regular language can be
 obtained as control language of labeled  flat finite splicing systems. 
Also, any non-empty context-free language can be obtained as control language by these systems 
and any recursively enumerable language can be obtained as
 control language when some rules are associated with label $\lambda$. 
\section{Control languages of labeled flat splicing systems}
\label{control languages of labeled flat splicing systems}

Let $\mathscr{S} = (A, I, R)$ be a flat splicing system. A labeled flat splicing system
 is a construct of the form $\mathscr{LS} = (A, I, R, Lab)$ 
such that $A \cap Lab = \emptyset$. Also, each rule of the flat splicing system are associated with 
a label from the set $Lab$. Unlike in the case of Szilard languages of labeled flat  
splicing systems, if the rules in $R$ are not labeled in one-to-one manner ( 
multiple rules can have the same label but one rule cannot have multiple labels) and 
also the rules can have empty label ( i.e., $\lambda$- label),
then the language obtained
by  concatenating the
labels of the applied rules of any terminal
derivation of the labeled flat splicing system  is called  as \textit{control language}.
The control languages of the labeled flat splicing system $\mathscr{LS}$ of type $(m, n)$ 
is denoted by 
$CL_{n, FAM}^{m}(\mathscr{LS})$.
The families of control languages 
of the labeled flat splicing systems of type 
$(m, n)$, 
 is denoted as $CLLS_{n, FAM}^{m}$. 

If the rules of $\mathscr{LS}$ are associated with label `$\lambda$'(empty) label, then  the control language of
the labeled  
flat splicing system $\mathscr{LS}$ of type $(m, n)$ is denoted by $CL_{n, \lambda, FAM}^{m}(\mathscr{LS})$. The families of control 
languages 
of the labeled flat splicing systems of type $(m, n)$ 
with  $\lambda$-labeled rules  is denoted 
by $CLLS_{n, \lambda, FAM}^{m}$.
When $m$ and $n$ are not specified, they are replaced by $``*".$

Now we give examples of labeled   flat  finite splicing systems which can obtain  non context-free and non regular languages 
as control languages.

\begin{example}

Let  $\mathscr{LS} = (A, I, R, Lab)$ be a labeled flat finite splicing system where
$A = \{X, Y, A_1, A^{'}\}$,
 $I = \{XY, A_1, A^{'}\}$,
$R = \{ a:<X ~|~ A_1 - \epsilon ~|~ Y>, a : <A_1 ~|~ A_1-\epsilon ~|~ Y>, b : <A_1 ~|~ A^{'} - \epsilon ~|~ Y>, b : <A_1 ~|~ A^{'} - \epsilon ~|~ A_1 A^{'}> \}.$
On application of the $a$- rules, one $A_1$ is added between the markers $X$ and $Y$. Similarly, the first $b$- rule  
insert  $A^{'}$ between $A$ and $Y$ and the second $b$- rule inserts $A^{'}$ between $A_1$ and $A_1A^{'}$. 
Also, the $b$-rules are applicable  after application of $a$-rules and only after application of same number of $a$ and $b$-rules
a string over $A$ is generated where no rule is applicable further.  Hence $CL_{1, FIN}^{2}(\mathscr{L S}) = \{ a^{n} b^{n} ~ |~ n \geq 1 \}.$
\end{example}

\begin{example} 

Let $\mathscr{LS} = (A, I, R, Lab)$ be a labeled flat  finite splicing system where
 $A = \{X, Y, A^{'}, B^{'}, A_1\}$,
 $I = \{XY, A^{'}, B^{'}, A_1\}$,
$R = \{ a : <X ~|~ A_1 - \epsilon ~|~ Y >, a : < A_1 ~|~ A_1 - \epsilon ~|~ Y>, 
b : < A_1 ~|~ A^{'} - \epsilon ~|~ Y>, b : < A_1 ~|~ A^{'} - \epsilon ~|~ A_1A^{'}>, c : <X A_1 A_1^{'} ~|~ B^{'} - \epsilon ~|~ A_1>, 
c : <B^{'} A_1 A_1^{'} ~|~ B^{'} - \epsilon ~|~ A_1>, c : <B^{'} A_1 A_1^{'} ~|~ B^{'} - \epsilon ~|~ Y> \}.$
Hence $CL_{1, FIN}^{3}(\mathscr{LS}) = \{a^{n}b^{n}c^{n} | n \geq 1\}.$
\end{example}

\section{Results}
\label{results1}

In the previous section, we proved that there exists some regular languages which cannot be obtained as Szilard
language by any labeled flat  finite splicing systems. But in the next result, we prove that any non-empty regular language can be 
obtained as  control language of  labeled  flat finite splicing systems of type $(1, 2).$

\begin{theorem}

$(REG \setminus \{\lambda\}) \subseteq CLLS_{2, FIN}^{1}.$
\end{theorem}
\begin{proof}
Let $L$ be a $\lambda$-free regular language and there exists a grammar $G = ( N, T, S, P)$ such that $L = L(G)$. 
The non-terminals $N$ of $G$ are renamed as $D_{i}, 1 \leq i \leq n$,  
starting from  $D_{1} = S$. 
 Now, the rules in $P$ are of the form $D_{i} \rightarrow a D_{i}$, $D_{i} \rightarrow a D_{j} ( i \neq j )$,
and $D_{i} \rightarrow a$, $D_{i}, D_{j} \in N$, and $a \in T$. 
In this proof we construct a  labeled flat splicing systems $\mathscr{LS}$ such that $L = L(G) = CL_{2, FIN}^{1}(\mathscr{LS})$.

\noindent  Let $\mathscr{LS} = (A,  I, R, Lab)$ be a labeled flat splicing system where:
  \begin{description} 
\item[$\bullet$]  $A = \{ X, Y, D_{1}, D_{2}, \ldots, D_{n}\} \cup \{Y_{a}\};$ 

\item[$\bullet$]$I = \{ XD_{1}Y\} \cup \{Y_{a} D_{i} | D_{i} \rightarrow a D_{i} \in P\} \cup \{Y_{a} D_{j} | D_{i} \rightarrow a D_{j} \in P\}
 \cup \{ Y_{a} | D_{i} \rightarrow a \in P\}$; 

 \item[$\bullet$]The rules in $R$  are of the following form: \\

$a: <D_{i} | Y_{a}- D_{i} | Y>$ for $D_{i} \rightarrow a D_{i}, D_{i} \in N, a \in T;$

$a: <D_{i} | Y_{a}- D_{j} | Y>$ for $D_{i} \rightarrow a D_{j}, D_{j} \in N, a \in T;$

$a: <D_{i}|\epsilon - Y_{a}| Y>$ for $D_{i} \rightarrow a, a \in T$;
\item[$\bullet$] $Lab = \{a ~|~  D_{i} \rightarrow a D_{i}, D_{i} \in N, a \in T\} \cup \{a ~|~  D_{i} \rightarrow a D_{j}, D_{j} \in N, a \in T\}
 \cup \{a ~|~ D_{i} \rightarrow a, a \in T\}$.
 \end{description}

 Recursive rules $D_{i} \rightarrow aD_{i}, D_{i} \rightarrow aD_{j}$ in $G$  are simulated by $a$- rules $a: <D_{i} | Y_{a}- D_{i} | Y>$ and 
 $a: <D_{i} | Y_{a}- D_{j} | Y>$  respectively. The terminal rule  $D_{i} \rightarrow a$ is simulated 
 by $a: <D_{i}|\epsilon - Y_{a}| Y>$. 
Hence $L \subseteq CL_{2, FIN}^{1}(\mathscr{LS}).$ 
Again, from  the one-to-one correspondence between the rules in $P$ and labeled flat splicing rules 
we have the inclusion $CL_{2, FIN}^{1}(\mathscr{LS}) \subseteq L(G)$.


\end{proof}

Now we show that any non-empty context-free language can be obtained as a control language of the
labeled  flat finite splicing systems of type $(2,*)$.

\begin{theorem}
$(CF \setminus \{\lambda\}) \subseteq CLLS_{2, FIN}^{2}.$
\end{theorem}
\begin{proof}
Let $L $ be a non-empty context-free language and  $G=(N,T, S, P)$ be a grammar  in
Greibach normal form such that $L = L(G).$ The rules in $P$ are of the form, 
$A_1 \rightarrow a \alpha$ and $A_1 \rightarrow a$, where $A_1 \in N, \alpha \in N^{+}, a \in T$.

The main idea of the proof is to construct a labeled flat splicing system
$\mathscr{LS} = ( A, I, R, Lab)$ 
such that $L = CL_{2, FIN}^{2}( \mathscr{LS})$ 
where $CL_{2, FIN}^{2}(\mathscr{LS})$ denotes the control language of the labeled flat splicing systems
 $\mathscr{LS}$ of type $(2,2)$.

At first the rules in $G$ are rewritten in the following manner:

$(1)$ If the terminal symbol at the starting of the right hand side of any two rules are same, i.e., say $a \in T$ is the starting 
symbol of any of the two rules. Then the $a \in T$ in 
 each rule is replaced with  distinct $a_{i}, i \in \mathbb{N}$. More precisely, for each distinct pair of rules
$A_1 \rightarrow a \alpha$ and $B_1 \rightarrow a \beta$ where $A_1, B_1 \in N, \alpha, \beta \in N^{+}$,
  are rewritten as:
$A_1 \rightarrow a_{k} \alpha$ and $B_1 \rightarrow a_{l} \beta$ where $k \neq l, k, l \in \mathbb{N}$.

$(2)$ Each distinct pair of rules of the form $A_1 \rightarrow a$ and $B_1 \rightarrow a$
 where $A_1, B_1 \in N$, are rewritten as $A_1 \rightarrow a_{k}$ 
and $B_1 \rightarrow a_{l}$ where $k \neq l, k, l \in \mathbb{N}.$

$(3)$ Also, if 
only one rule  $A_1 \rightarrow a \alpha, A_1 \in N, \alpha \in N^{+}, a \in T$
is present in $G$, then the rule is rewritten as $A_1 \rightarrow a_{1} \alpha.$ Similarly, if there exists only one rule 
$A_1 \rightarrow a, A_1 \in N, a \in T$, then it is rewritten as $A_1 \rightarrow a_{1}.$

Now, we construct a labeled flat splicing system which simulates the newly transformed rules. 

 Let $\mathscr{LS} = ( A,  I, R, Lab)$ be a labeled flat splicing system
where: 
\begin{itemize}
\item[$\bullet$] $A = \{X, Y \} \cup N \cup \{ Y_{a_{i}} ~|~ A_1 \rightarrow a_{i} \alpha \in P\} \cup  \{ Y_{a_{i}} ~|~  A_1 \rightarrow a_{i} \in P\}$;
\item[$\bullet$] $I = \{ XSY \} \cup \{Y_{a_{i}} \alpha | A_1 \rightarrow a_{i} \alpha \in P\}
 \cup \{ Y_{a_{i}} | A_1 \rightarrow a_{i} \in P\}$;


\item[$\bullet$]
$R$ contains the following rules: 

For $A_1 \rightarrow a_{i} \alpha$:

$a: <XS | Y_{a_{i}}- \beta | Y>$ where $Y_{a_{i}} - \beta = Y_{a_{i}} \alpha \in I, \beta \in N, a_i \in T$;

$a: <Y_{a_{j}} \alpha_1 | Y_{a_{i}}- \beta | \alpha_2>$, where $\alpha_1 \in N, \alpha_2 \in N \cup \{Y\}, Y_{a_{i}} - \beta = Y_{a_{i}} \alpha \in I, 
\beta \in N, a_i, a_j \in T$.


For  $A_1 \rightarrow a_{i}$:

$a: <XS | \epsilon-Y_{a_{i}} | Y>, a_i \in T$;

$a: <Y_{a_{j}} \alpha_1 | \epsilon-Y_{a_{i}} | \alpha_2>, \alpha_1 \in N, \alpha_2 \in N \cup \{Y\}, a_i, a_j \in T$
where $i, j \in \mathbb{N}$.

\item[$\bullet$] $Lab = \{ a | A_1 \rightarrow a_{i} \alpha \} \cup \{ a | A_1 \rightarrow a_{i}\}$;

\end{itemize}

 
We first prove that $L(G) = L \subseteq CL_{2, FIN}^{2}( \mathscr{LS})$. Any element $x \in L$ can be generated after sequential 
application of the rules in $P$. If the flat splicing rules are applied in the
same sequence starting from $XSY$, a string over $A_{1}$ is generated where 
no rule  can be applied further, i.e., a terminal derivation is obtained. Also, concatenation of the labels 
of the flat splicing rules generates the string $x$. Hence  $L(G) = L \subseteq CL_{2, FIN}^{2}( \mathscr{LS})$. 

In the similar manner,  we can prove the other inclusion $CL_{2, FIN}^{2}(\mathscr{LS}) \subseteq L(G) = L$. 
 Let $w_{1}  \in CL_{2, FIN}^{2}(\mathscr{LS})$, i.e., there exists a terminal
 derivation in $\mathscr{LS}$ such that the concatenation of the labels of the splicing rules generate $w_{1}$.
   If the rules in $G$ are applied in the same sequence as in the terminal derivation of $\mathscr{LS}$ generating $w_{1}$, 
   the string $w_{1}$ is generated.
    This will imply $CL_{2, FIN}^{2}(\mathscr{LS})) \subseteq L.$

\end{proof}

In the next theorem, we show that if the rules in $\mathscr{LS}$ are labeled with $\lambda$, then any recursively 
enumerable language can be obtained as a control language by the labeled flat finite splicing systems of type $(4, 2).$
\begin{theorem}

$RE = CL_{\lambda}LS_{2, FIN}^{4}$.
\end{theorem}
\begin{proof}
The inclusion $CL_{\lambda}LS_{2, FIN}^{4} \subseteq RE$ follows from the Church-Turing thesis. It only remains to
prove the inclusion $RE \subseteq CL_{\lambda}LS_{2, FIN}^{4}$.
The proof of this inclusion follows from the proof of the Theorem $12$. If all the labels of the rules of Theorem $12$ except
$a_{i}^{1}, a_{i}^{2}, a_{i}^{3}, a_{i}^{4}, a_{i}^{5}, a_i^{6}$ and $a_i^{7}$  are replaced with $\lambda$ and the $a_{i}^{1}, a_{i}^{2}, a_{i}^{3}, a_{i}^{4}, a_{i}^{5}, a_i^{6}$ and $a_i^{7}$ labeled rules are 
replaced by $a$, then by following the same procedure 
as in Theorem $12$ we can prove that $RE \subseteq CL_{\lambda}LS_{2, FIN}^{4}$.

\end{proof}

\section{Conclusion} 
\label{conclusion}
In this work, we compared the Szilard and control languages of labeled flat splicing systems 
with the family of 
languages in the Chomsky hierarchy. It has been shown that 
Szilard language of labeled  flat finite splicing systems and family of regular, context-free and context-sensitive languages are incomparable. Also, any 
non-empty regular and context-free language
 can be obtained as Szilard 
language of these systems when a homomorphism is applied.
Also any
 recursively enumerable language can be obtained as homomorphic image of Szilard language of 
 labeled  flat finite splicing system of type $( 4, 2)$.   
We also proved that
any non-empty regular and context-free language can be obtained as control language
by labeled  flat finite splicing systems and any recursively enumerable language 
can be obtained as control language if some of the rules can be labeled with
the empty label (i.e., $\lambda$). It remains to be investigated whether the bounds mentioned in this paper are optimal.

\section{Acknowledgement}

First author acknowledges the fund received from Indian Statistical Institute, Kolkata, West Bengal, India.

\section{References}




\begin{thebibliography}{9}



\bibitem{berstel} J. Berstel, L. Boasson and I. Fagnot, Splicing systems and the Chomsky hierarchy, {\it Theor. Comp. Sci}. {\bf 436}(2012), 2-22.


\bibitem{kgs-2} R. Ceterchi, L. Pan, B. Song and K. G. Subramanian, Matrix flat splicing systems, 
{\it (BIC-TA 2016)}, CCIS {\bf  681}(2017), 54-63.

\bibitem{culik} K.  Culik II and T. Harju, Splicing semigroups of dominoes and DNA, {\it Dis. App. Math.} {\bf 31}(1991) 261-277. 

\bibitem{makinen-2} L. Cojocaru and E. M{\" a}kinen, On some derivation mechanisms and the complexity of their Szilard languages, {\it Theor. Comp. Sci.}
{\bf 537}(2014), 87-96.

\bibitem{fkp} 
R. Freund, L. Kari and G. P{\u a}un, DNA computing based on splicing : the existence of universal computers, {\it Th. of Comp. Syst}. {\bf 32}(1999)69--112.



\bibitem{prithwi-1} K. Mahalingam , P.  Paul and E. M{\" a}kinen, On derivation languages of a class of splicing systems, {\it Act. Cybern.} {\bf 23}(2018) 1--13.

\bibitem{prithwi-2} K. Mahalingam, P. Paul, B. Song, L. Pan, K. G. Subramanian, Derivation languages of Splicing P systems, {\it BIC-TA 2017}, CCIS {\bf 791}(2017),
487-501.

\bibitem{reproof} C. Martin-Vide, G. P{\u a}un, A. Salomaa, Characterizations of recursively enumerable languages by means of insertion grammars, {\it Theo. Comput. Sci.}. {\bf 2005} (1998), 195-205.

\bibitem{dsp} 
 G. P{\u a}un, DNA computing based on splicing : universality results, {\it Theor. comp. sci}. {\bf 231}(2000) 275--296.
 \bibitem{tp} 
 G. P{\u a}un, Splicing systems with targets are computationally universal, {\it Inf. Pro. Lett.} {\bf 59(3)}(1996) 129--133.

\bibitem{dnabook} G. P{\u a}un, G. Rozenberg and A. Salomaa, {\it DNA Computing: New Computing Paradigms}, (Springer-Verlag 1998).


\bibitem{Paun} G. P{\u a}un, On some families of Szilard languages, {\it Bull. math. de la Soc. des Sci. Math. de la Rép.
 Soc. de Roum}.
 {\bf 27(75)}(1983), 259-265.

\bibitem{kgs-1} L. Pan, B. Song, A. K. Nagar and K. G.  Subramanian, Language generating alphabetic flat splicing 
P systems. {\it Theor. Comp. Sci.} (2017), \url{doi.org/10.1016/j.tcs.2017.12.014}.



\bibitem{handbook} G. Rozenberg and A. Salomaa (eds.),  {\it Handbook of formal languages}, Vol. 1 (Springer, Berlin, 1997).

\bibitem{ajeesh-1} A. Ramanujan and K. Krithivasan, Control Words of Transition P Systems, 
{\it (BIC-TA 2012)}, Adv. in Intell. Syst. and Comp {\bf  201}(2013), Springer, India (2013)
\bibitem{ajeesh-2} A. Ramanujan and K., Krithivasan, Control Languages Associated with Tissue P Systems.
 {\it UCNC 2013}, LNCS {\bf  7956}(2013), Springer, Berlin, Heidelberg.

\bibitem{ajeesh-3} A. Ramanujan and K. Krithivasan, Control languages associated with Spiking Neural P systems,
{\it Rom. Jour. of Infor. Sci. and Tech.}. {\bf 15(4)} (2012), 301-318.



\bibitem{kgs-3} G. Samdanielthompson, N. G. David and K. G. Subramanian, Alphabetic flat splicing pure context-free
grammar systems, {\it Jour. of Math. and Infor.}. {\bf 7} (2017), 1-6.

\bibitem{kgs-4} G. Samdanielthompson, N. G. David, A. K. Nagar and K. G. Subramanian, Flat splicing array 
grammar systems generating picture arrays, {\it  Int. Jour. of Compt. Infor. Syst. and 
Indust. Manag. Appl.}, {\bf 8}(2016), 336-344. 

\bibitem{pan-1} X. Zhang, Y. Liu, B. Luo and L. Pan, Computational power of tissue P systems for generating control languages.
{\it Infor, Sci.}. {\bf 2014} (2014), 285-297.


\end{thebibliography}


\begin{center}
\textbf{Appendix}

\end{center}

\textbf{Extended proof of Theorem \ref{rehomomorphism}:}

\textbf{Simulation of $r_i: A_1 \rightarrow B_1C_1$ using $r_i^{1}, r_i^{2}, r_i^{3}, r_i^{4}, r_i^{5}$-rule:}

For each $r_i: A_1 \rightarrow B_1C_1$ there exists $r_i^{1}, r_i^{2}, r_i^{3}, r_i^{4}$ and $r_i^{5}$-rule.  These rules 
are applicable to the word $Xw_1 A_1w_2Y$ where $w_1, w_2 \in A^{*}$. Moreover, after application of the rules depending
on the contexts, the words $Xw_1 A_1w_2Y$ and $[r_i]B_1C_1$ are spliced together and the word  $Xw_1A_1[r_i]B_1C_1w_2Y$ is obtained. 

$(Xw_1 A_1w_2Y, [r_i]B_1C_1) \vdash^{r_i^{1}, r_i^{2}, r_i^{3}, r_i^{4}, r_i^{5}}  Xw_1A_1[r_i]B_1C_1w_2Y$  \hspace {2cm}\ldots (1)

No rule of $R$ is applicable to the subword $A_1[r_i]$. But rules from $(R_{15})$ are applicable to the subword $[r_m]A[r_i]$
 obtained during identification of the leftmost non-terminal symbol for simulation of terminating rules.

\textbf{Simulation of $r_i: A_1B_1 \rightarrow C_1D_1$ using $r_i^{7}, r_i^{8}, r_i^{9}$-rule:}

For each $r_i: A_1B_1 \rightarrow C_1D_1$ there exists $r_i^{7}, r_i^{8}$ and $r_i^{9}$-rule. 
These rules splice the strings $Xw_1A_1B_1 \alpha_1 \alpha_2 w_2Y$ and $[r_i]C_1D_1$ where $\alpha_1, \alpha_2 \in N, w_1, w_2 \in A^{*}$.

$(Xw_1A_1B_1 \alpha_1 \alpha_2 w_2Y, [r_i]C_1D_1) \vdash^{r_i^{7}} Xw_1A_1B_1 [r_i] C_1 D_1 \alpha_1 \alpha_2 w_2 Y$ \hspace {0.7cm} \ldots (2)

$(Xw_1A_1B_1 \alpha_1 Y, [r_i]C_1D_1) \vdash^{r_i^{9}} Xw_1A_1B_1 [r_i] C_1 D_1 \alpha_1 Y$ \hspace {2.4cm} \ldots (3)

$(Xw_1A_1B_1  Y, [r_i]C_1D_1) \vdash^{r_i^{8}} Xw_1A_1B_1 [r_i] C_1 D_1  Y$ \hspace {3.2cm} \ldots (4)

No rule in $R$ is applicable to the subword $A_1B_1[r_i]$ of the words obtained in $(2), (3)$ and $(4)$. 
But the rules in $(R_{15})$ is applicable to the subword 
$[r_m]A_1B_1[r_i]$ is obtained during the computation. 

\textbf{Simulation of $r_i: A_1B_1 \rightarrow C_1D_1$ in a word of the form $Xw_1A_1 \alpha_1[r_1]$
$ \alpha_2 [r_2] \ldots \alpha_n[r_n] B_1w_2Y$
where $\alpha_1, \ldots, \alpha_n \in N, [r_1], \ldots, [r_n] \in \Delta_1, w_1, w_2 \in A^{*}$ using $r_i^{10}, r_i^{11}$ and $r_i^{12}$-rule:}

$(Xw_1A_1 \alpha_1[r_1]\alpha_2[r_2] \ldots \alpha_n [r_n]B_1w_2 Y, [r_i])$ 

$\vdash^{r_i^{10}}$ $Xw_1A_1 [r_i] \alpha_1[r_1]\alpha_2[r_2] \ldots \alpha_n [r_n]B_1w_2 Y$

$(Xw_1A_1 [r_i] \alpha_1[r_1]\alpha_2[r_2] \ldots \alpha_n [r_n]B_1w_2 Y, [r_i])$ 

$\vdash^{r_i^{11}}$$ Xw_1A_1 [r_i] \alpha_1[r_1] [r_i]\alpha_2[r_2] \ldots \alpha_n [r_n]B_1w_2 Y$

\vdots

$(Xw_1A_1 [r_i] \alpha_1[r_1] [r_i]\alpha_2[r_2] \ldots [r_i] \alpha_n [r_n]B_1w_2 Y, [r_i])$ 

$\vdash^{r_i^{11}}$ $ Xw_1A_1 [r_i] \alpha_1[r_1] [r_i]\alpha_2[r_2] \ldots [r_i] \alpha_n [r_n] [r_i]B_1w_2 Y$

$(Xw_1A_1 [r_i] \alpha_1[r_1] [r_i]\alpha_2[r_2] \ldots [r_i] \alpha_n [r_n] [r_i]B_1w_2 Y, [r_i]C_1D_1)$ 

$\vdash^{r_i^{12}}$ $Xw_1A_1 [r_i] \alpha_1[r_1] [r_i]\alpha_2[r_2] \ldots [r_i] \alpha_n [r_n] [r_i]B_1 [r_i] C_1 D_1w_2 Y$ \hspace {1.1cm} \ldots (5)

\textbf{Application of $a_i^{1}$-rule and $r_j^{14}$-rule:}

The $a_i^{1}$-rule for each $r_i: A_1 \rightarrow a$ and $r_j^{14}$-rule for each $r_j: A_2 \rightarrow \lambda$
is applicable to the word $XA_1Y$ and $XA_2Y$.

$(XA_1Y, k_a^{i}) \vdash^{a_i^{1}} XA_1k_a^{i}Y$ \hspace {7.3cm} \ldots (6) 

$(XA_2Y, k_{\lambda}^{i}) \vdash^{r_j^{14}} XA_2k_{\lambda}^{i}Y$ \hspace {7.3cm} \ldots (7)

\textbf{Simulation of $r_i: A_1 \rightarrow a$ using the $ a_i^{2}, a_i^{3}, a_i^{4}, a_i^{5}, a_i^{6}$-rule:}

The flat splicing rules simulating the terminal rules $r_i: A_1 \rightarrow a$ and $r_i: A_1 \rightarrow \lambda$ are constructed 
in such a manner that the left-most non-terminal is identified by the marker $X$ and the symbol $[r_m]$ and then the splicing is performed. 
This process is performed by application of the $r_{m+1}^{i}, r_{m+2}^{i}$ -rule and the rules in $(R_{15})$.  
At first we discuss the simulation of the rules $r_i: A_1 \rightarrow a$ and $r_j: A_1 \rightarrow \lambda$. 

\textbf{Simulation of $r_i: A_1 \rightarrow \lambda$ using the $ a_i^{2}, a_i^{3}, a_i^{4}, a_i^{5}, a_i^{6}$-rule:}

For each $r_i: A_1 \rightarrow a$ there exists rules $a_i^{1}, a_i^{2}, a_i^{3}, a_i^{4}, a_i^{5},$ and  $a_i^{6}$-rule  which 
splice the words $Xw_1A_1w_2Y$ and $k_a^{i}$. 

$(Xw_1[r_m]A_1Y, k_a^{i}) \vdash^{ a_i^{2}} Xw_1A_1k_a^{i}Y$

$(Xw_1[r_m]A_1 \alpha_1 Y, k_a^{i}) \vdash^{ a_i^{3}} Xw_1A_1k_a^{i}\alpha_1Y, \alpha_1 \in N$

$(Xw_1[r_m]A_1 \alpha_1 \alpha_2 Y, k_a^{i}) \vdash^{ a_i^{4}} Xw_1A_1k_a^{i}\alpha_1 \alpha_2Y, \alpha_1, \alpha_2 \in N$

$(Xw_1[r_m]A_1 \alpha_1 \alpha_2 \alpha_3 Y, k_a^{i}) \vdash^{ a_i^{5}} Xw_1A_1k_a^{i}\alpha_1 \alpha_2 \alpha_3Y, \alpha_1, \alpha_2, \alpha_3 \in N$,

$(Xw_1[r_m]A_1 \alpha_1 \alpha_2 \alpha_3 \alpha_4 Y, k_a^{i}) \vdash^{a_i^{6}} Xw_1A_1k_a^{i} \alpha_1 \alpha_2 \alpha_3 \alpha_4Y$ 

 \hspace {11.6 cm}  \ldots (8)


Similarly, the rule $r_i: A_1 \rightarrow \lambda$ can be simulated by application of the $ r_i^{15}, r_i^{16}, r_i^{17}, r_i^{18},$ and $r_i^{19}$-rule.


\textbf{Application of  $r_{m+1}^{i}$-rule for $r_i: A_1 \rightarrow a$ and $r_{m+1}^{j}$-rule for $r_j: A_2 \rightarrow \lambda$:}

After simulating the rules $r_i: A_1 \rightarrow a$ and $r_j: A_2 \rightarrow  \lambda$ the words $[r_m]A_1k_a^{i}$ and 
$[r_m]A_2k_{\lambda}^{j}$ are obtained after application of the rules $a_i^{2}, \ldots, a_i^{6}$ and $r_i^{15}, \ldots, r_i^{19}$ 
respectively. To proceed the computation further  $r_{m+1}^{i}$-rule and $r_{m+2}^{j}$-rule are applied. 
These rules are applicable to the words $Xw_1[r_m]A_1k_a^{i} \alpha_1 \alpha_2 w_2Y$ and  $Xw_1[r_m]A_2k_{\lambda}^{j} \alpha_1 \alpha_2 w_2Y$
respectively where $w_1, w_2$
$ \in A^{*}, \alpha_1 \in N, \alpha_2 \in \{Y\} \cup N \cup \Delta_1 \cup \Delta_2$. 
This computation can be done by splicing the words $Xw_1[r_m]A_1k_a^{i} \alpha_1 \alpha_2 w_2Y$ and  $Xw_1[r_m]A_2k_a^{j} \alpha_1 \alpha_2 w_2Y$
with $[r_m]$.

$(Xw_1[r_m]A_1k_a^{i} \alpha_1 \alpha_2 w_2Y, [r_m]) \vdash^{r_{m+1}^{i}} Xw_1[r_m]A_1k_a^{i} [r_m] \alpha_1 \alpha_2 w_2Y$  \hspace {.6 cm} \ldots (9)

$(Xw_1[r_m]A_2k_{\lambda}^{j} \alpha_1 \alpha_2 w_2Y, [r_m]) \vdash^{r_{m+2}^{j}} Xw_1[r_m]A_2k_{\lambda}^{j} [r_m] \alpha_1 \alpha_2 w_2Y$  
 \hspace {.6 cm}\ldots (10) 

The symbol $[r_m]$ helps the system to identify the leftmost non terminal where the rules  $a_i^{2}, \ldots, a_i^{6}$ and $r_i^{15}, \ldots, r_i^{19}$ 
can be applied. The identification process is performed by the $r_m, r_{m+1}, r_{m+2}, r_{m+3}, r_{m+4}$ and $r_{m+5}$-rule.

\textbf{Application of $r_m, r_{m+1}, r_{m+2}, r_{m+3}, r_{m+4}, r_{m+5}, r_{m+6}$-rule:}

Any computation of $\mathscr{LS}$ starts with $XSY$ and a rule $r_i: A_1 \rightarrow B_1C_1$ 
is applied and the word $XA_1[r_i]B_1C_1Y$ is obtained. 
During the computation since no words can be removed from the word $XwY, w \in A^{+}$, 
the word present in the system will be of the form $XA_1[r_i] \alpha w_2Y$ where $\alpha \in N, w_2 \in A^{+}$. 
Moreover, the terminal rules are applied in leftmost manner and the corresponding flat 
splicing rules are constructed in such a manner such that the rules are applied to the leftmost 
non terminal.
So, to identify the leftmost non terminal $r_m$-rule is applied  first. 

$(XA[r_i]\alpha w_2Y, [r_m]) \vdash^{r_m} XA[r_i][r_m] \alpha w_2 Y$ \hspace{4cm} \ldots (11)

Moreover, after simulation of the rule $r_i: A_1B_1 \rightarrow C_1D_1$ using the $r_i^{7}, r_i^{8}$ and $r_i^{9}$-rule
we have a word $Xw_1A_1B_1[r_i]w_2Y, w_1, w_2 \in A^{+}$. The subword $A_1B_1[r_i]$ becomes inactive, i.e., no rule 
is applicable to the subword. But it becomes active again when the subword $[r_m]A_1B_1[r_i]$ is obtained.

$(Xw_1[r_m]\alpha_k[r_k] A_1B_1[r_i] w_2 Y, [r_m])$ 

$\vdash^{r_{m+3}} 
Xw_1[r_m]\alpha_k[r_k] [r_m] A_1B_1[r_i] w_2 Y$

$Xw_1[r_m]\alpha_k[r_k] [r_m] A_1B_1[r_i] w_2 Y$

$\vdash^{r_{m+1}} 
Xw_1[r_m]\alpha_k[r_k] [r_m] A_1B_1[r_i] [r_m] w_2 Y$  \hspace{4cm} \ldots (12)

After the simulation of $r_i: A_1B_1 \rightarrow C_1D_1$ in the word $Xw_1 \alpha [r_1] \alpha_2 [r_2] \ldots$
$ \alpha_n [r_n] w_2Y$, 
to further proceed the computation the following rules are applied. 
Moreover, after simulation of the rule $r_i: A_1B_1 \rightarrow C_1D_1$, the word $Xw_1A_1[r_i] \alpha_1[r_1] [r_i] \alpha_2 [r_2] [r_i] \ldots[r_i] \alpha_n [r_n][r_i]B_1[r_i]C_1D_1w_2Y$ where $w_1, w_2 \in A^{+}$ is obtained. 
So, the computation proceeds in the following manner when a word $Xw_1[r_m]$
$A_1[r_i] \alpha_1$
$[r_1] [r_i] \alpha_2 [r_2] [r_i] \ldots[r_i] \alpha_n [r_n][r_i]B_1[r_i]w_2Y$  is obtained. 

$(Xw_1[r_m]A_1[r_i] \alpha_1$
$[r_1] [r_i] \alpha_2 [r_2] [r_i] \ldots[r_i] \alpha_n [r_n][r_i]B_1[r_i]w_2Y, [r_m])$

$ \vdash^{r_{m+2}}$
$Xw_1[r_m]A_1[r_i]  [r_m] \alpha_1$
$[r_1] [r_i] \alpha_2 [r_2] [r_i] \ldots[r_i] \alpha_n [r_n][r_i]B_1[r_i]w_2Y$;
 
$Xw_1[r_m]A_1[r_i]  [r_m] \alpha_1$
$[r_1] [r_i] \alpha_2 [r_2] [r_i] \ldots[r_i] \alpha_n [r_n][r_i]B_1[r_i]w_2Y$

$ \vdash^{r_{m+3}}$
$Xw_1[r_m]A_1[r_i]  [r_m] \alpha_1$
$[r_1] [r_m] [r_i] \alpha_2 [r_2] [r_i] \ldots[r_i] \alpha_n [r_n][r_i]B_1[r_i]$
$w_2Y$;

So repeated application of the $r_{m+3}$-rule the word   $Xw_1[r_m]A_1[r_i]  [r_m] \alpha_1$
$[r_1] [r_m] [r_i] \alpha_2 [r_2] [r_i] \ldots[r_i] [r_m] \alpha_n [r_n][r_i]B_1[r_i]w_2Y$ 
is obtained. 

$(Xw_1[r_m]A_1[r_i]  [r_m] \alpha_1$
$[r_1] [r_m] [r_i] \alpha_2 [r_2] [r_i] \ldots[r_i] [r_m] \alpha_n [r_n][r_i]B_1[r_i]w_2Y,$

$ [r_m])$

$ \vdash^{r_{m+4}}$
$Xw_1[r_m]A_1[r_i]  [r_m] \alpha_1$
$[r_1] [r_m] [r_i] \alpha_2 [r_2] [r_i] \ldots[r_i] [r_m] \alpha_n [r_n] [r_i] [r_m] B_1$
$[r_i]w_2Y$;

$(Xw_1[r_m]A_1[r_i]  [r_m] \alpha_1$
$[r_1] [r_m] [r_i] \alpha_2 [r_2] [r_i] \ldots[r_i] [r_m] \alpha_n [r_n] [r_i] [r_m] B_1[r_i]w_2Y,$
$ [r_m])$

$ \vdash^{r_{m+6}}$
$Xw_1[r_m]A_1[r_i]  [r_m] \alpha_1[r_1] [r_m] [r_i] \alpha_2 [r_2] [r_i] \ldots[r_i] [r_m] \alpha_n [r_n] [r_i] [r_m] B_1$

$[r_i] [r_m] C_1D_1w_2Y$;  \hspace{8cm} \ldots (13)

\textbf{Application of the rule $r_i^{6}$ for $r_i: A_1B_1 \rightarrow C_1D_1$:}

After the simulation of $r_i: A_1B_1 \rightarrow C_1D_1$, a word $Xw_1A_1[r_i] \alpha_1[r_1] [r_i]$
$ \ldots \alpha_n[r_n][r_i]B_1[r_i]w_2Y, w_1, w_2 \in A^{+} $
is obtained. If $w_1 = w_1^{'}A_2, w_1^{'} \in A^{*}$ and there exists a rule $r_j: A_2 \rightarrow B_2C_2$, then the application 
of the rule is simulated in the following manner

($Xw_1^{'}A_2A_1[r_i] \alpha_1[r_1] [r_i] \ldots \alpha_n[r_n][r_i]B_1[r_i]w_2Y, [r_j])$

$\vdash^{r_i^{6}}$
 $Xw_1^{'}A_2[r_j] B_2C_2[r_i] \alpha_1[r_1] [r_i] \ldots \alpha_n[r_n][r_i]B[r_i]w_2Y$.  \hspace{1.5cm} \ldots (14)

 \textbf{Application of $a_j^{7}$-rule for  $a_j: A_3 \rightarrow a$ and $r_l^{20}$-rule for  $r_l: A_4 \rightarrow \lambda$: }
 
 Similarly, the application of the $a_j^{7}$-rule for  $a_j: A_3 \rightarrow a$ and  $r_l^{20}$-rule  for $r_l: A_4 \rightarrow \lambda$
 can be explained for the words 
 $Xw_1^{'}A_3A_1[r_i] \alpha_1[r_1] [r_i] \ldots \alpha_n[r_n][r_i]$
 $B_1[r_i]w_2Y$  and 
 $Xw_1^{'}A_4A_1[r_i] \alpha_1[r_1] [r_i] \ldots \alpha_n[r_n][r_i]$
 $B_1[r_i]w_2Y$ respectively in the following manner: 
 
 $Xw_1^{'}A_3A_1[r_i] \alpha_1[r_1] [r_i] \ldots \alpha_n[r_n][r_i]B_1[r_i]w_2Y$
 
 $\vdash^{a_j^{7}}$
 $Xw_1^{'}A_3[r_j] k_a^{j}A_1[r_i] \alpha_1[r_1] [r_i] \ldots \alpha_n[r_n][r_i]B_1[r_i]w_2Y$  \hspace{1.8cm} \ldots (15)

  $Xw_1^{'}A_4A_1[r_i] \alpha_1[r_1] [r_i] \ldots \alpha_n[r_n][r_i]B_1[r_i]w_2Y$
 
 $\vdash^{r_l^{20}}$
 $Xw_1^{'}A_4[r_l] k_{\lambda}^{l}A_1[r_i] \alpha_1[r_1] [r_i] \ldots \alpha_n[r_n][r_i]B_1[r_i]w_2Y$  \hspace{1.8cm} \ldots (16)

From the above derivations we know that after simulating the rules, subwords of the form $A_1[r_i], A_1B_1[r_i], A_1k_a^{i}$
and $A_1k_{\lambda}^{i}$ are obtained. Moreover, these subwords become active again when the subwords 
$[r_m]A_1[r_i], [r_m]A_1B_1$
$[r_i], [r_m]A_1k_a^{i}$
and $[r_m]A_1k_{\lambda}^{i}$. So the flat splicing rules are constructed in such a manner that 
no extra derivation is possible. Again, any $x \in L$ can be generated
by application of the recursive rules and then by leftmost application of the terminating rules. 
Now, if the flat splicing rules simulating the rules in $P$ are applied in the same order, then 
there exists a terminal derivation such that the concatenation of the labels of the applied 
flat splicing rules obtain a string $x_1$ where $h(x_1) = x$.  Hence $x \in SZ_{2}^{4}(\mathscr{LS})$. 
So, $   L  \subseteq SZ_{2}^{4}(\mathscr{LS}) $.

 Again, let $x_1 \in SZ_{2}^{4}(\mathscr{LS}) $. Then $x_1$ can be obtained by concatenating the labels
 of the flat splicing rules in a terminal derivation. Since no extra derivation is possible, 
 if the rules in $P$ are applied in the same order, a string over $T$ is obtained. In fact, the rules
 are constructed in such a way that, the string $x \in T^{+}$ is obtained where $h(x_1) = x$. 
 Hence $  SZ_{2}^{4}(\mathscr{LS})   \subseteq  L $.




\end{document}